\documentclass[12pt]{article}
\usepackage[margin=1in]{geometry}

\usepackage{amsthm,amsmath,amsfonts,amssymb, mathtools}
\usepackage{enumerate, bbm}
\usepackage{graphicx}
\usepackage{algorithm}
\usepackage[noend]{algpseudocode}
\usepackage{authblk}
\usepackage{subfig}
\usepackage{multirow}
\usepackage{booktabs}
\usepackage{natbib}
\usepackage[colorlinks,citecolor=blue,urlcolor=blue]{hyperref}

\usepackage{setspace}

\newtheorem{theorem}{Theorem}

\newtheorem{lemma}{Lemma}
\newtheorem{assumption}{Assumption}

\newtheorem{proposition}{Proposition}

\newtheorem{remark}{Remark}

\newcommand{\Real}{\mathbb{R}}

\newenvironment{procedure}[1][]{\refstepcounter{example}\par\medskip
   \noindent \textbf{Procedure~\theexample #1.} \begin{itshape}}{\medskip\end{itshape}}

\newtheorem{definition}[theorem]{Definition}
\newtheorem{example}{\textsc{Example}}

\usepackage{xcolor}

\newcommand{\R}{\mathbb{R}}

\newcommand{\E}{\mathbb{E}}
\renewcommand{\P}{\mathbb{P}}

\newcommand{\YY}{{Y}}
\newcommand{\ZZ}{{Z}}
\newcommand{\XX}{\mathbf{X}}

\newcommand{\calN}{\mathcal{N}}

\newcommand{\Pn}{\mathbb{P}_n}

\newcommand{\bX}{\mathbf{X}}

\newcommand{\cX}{\mathcal{X}}

\newcommand{\cF}{\mathcal{F}}
\newcommand{\cI}{\mathcal{I}}

\newcommand{\cM}{\mathcal{M}}
\newcommand{\cZ}{\mathcal{Z}}

\newcommand{\cR}{\mathcal{R}}

\newcommand{\glob}{\textnormal{glob}}
\newcommand{\het}{\textnormal{het}}
\newcommand{\spill}{\textnormal{sp}}
\newcommand{\im}{\textnormal{imb}}

\newcommand{\mlfrt}{\texttt{ML-FRT}}
\newcommand{\lmfrt}{\texttt{LM-FRT}}
\newcommand{\res}{\mathrm{RES}}
\newcommand{\ml}{\mathrm{ML}}

\newcommand{\CV}{\mathrm{CV}_{n,k}}

\newcommand{\var}{\mathrm{var}}

\newcommand{\normal}{\mathcal{N}}

\newcommand{\Ind}{\mathbbm{1}}

\newcommand{\indep}{\perp \!\!\! \perp}
\newcommand\myeq[1]{\stackrel{\mathclap{\normalfont\mbox{\tiny #1}}}{=}}
\renewcommand{\epsilon}{\varepsilon}
\newcommand{\eps}{\epsilon}

\newcommand{\EX}{\mathbb{E}}
\newcommand{\VAR}{\mathbb{V}\mathrm{ar}}

\title{ML-assisted Randomization Tests for Detecting Treatment Effects in A/B Experiments}
\author{Wenxuan Guo\thanks{Booth School of Business, University of Chicago. \url{wxguo@chicagobooth.edu}}, 
JungHo Lee\thanks{Carnegie Mellon University. \url{junghol@andrew.cmu.edu}}, Panos Toulis\thanks{Booth School of Business, University of Chicago. \url{panos.toulis@chicagobooth.edu}. PT acknowledges support from NSF SES-2419009. All authors wish to thank Iavor Bojinov, Michael Hudgens, Dominik Rothenh\"{a}usler and the participants at the Online Causal Inference workshop for valuable feedback and comments.}}
\date{\today}

\begin{document}
\onehalfspacing

\maketitle

\begin{abstract}
Experimentation is widely utilized for causal inference and data-driven decision-making across disciplines. In an A/B experiment, for example, an online business randomizes two different treatments (e.g., website designs) to their customers and then aims to infer which treatment is better. 
In this paper, we construct randomization tests for complex treatment effects, including heterogeneity and interference. A key feature of our approach is the use of flexible machine learning (ML) models, where the test statistic is defined as the difference between the cross-validation errors from two ML models, one including the treatment variable and the other without it. This approach combines the predictive power of modern ML tools with the finite-sample validity of randomization procedures, enabling a robust and efficient way to detect complex treatment effects in experimental settings.  We demonstrate this combined benefit both theoretically and empirically through applied examples. 
\end{abstract}

\newpage
\doublespacing
\section{Introduction}
Experimentation lies at the heart of causal inference and data-driven decision-making. To improve sales, for instance, an online business may randomly display two different website designs, known as ``A/B experiment'', and choose the most effective design in generating revenue. In public policy, a governmental agency may randomly assign students to an after-school program to assess its effectiveness in educational outcomes. 

In recent years, flexible ML methods have been used for causal inference, following the principle that an effective treatment should be predictive of the outcome~\citep{granger1969investigating}. 
These {\em causal ML} methods aim to learn flexible models separately for the outcome and the treatment, and then adjust estimation to avoid bias. In the majority of cases, the target of inference is an average treatment effect under a framework of i.i.d. data. As a result, these methods often lack finite-sample guarantees and may not be suitable in settings where either the target of inference or the experimental design is complex.

In many practical settings, however, the goal is to {\em detect} whether a certain type of treatment effect is present in the experiment. 
Although it might appear simpler than quantifying the treatment effect, the problem of detecting complex treatment effects often presents distinct challenges, especially when the target relates to a complex interaction between the treatment and other variables in the system. For instance, suppose an online transportation platform, such as Uber, performed an A/B test on a new feature that allows users to add dependents in their account. One might ask: Does the new feature interact with the existing driver-customer matching algorithm? Furthermore, 
does incorporating, say, real-time traffic information influence this interaction? 
In such problems, the notion of an average treatment effect becomes ambiguous, especially when the effects are expected to be heterogeneous and involve high-order nonlinear interactions.

In this paper, we propose the combination of flexible ML models with 
{\em randomization inference} to detect complex treatment effects. The underlying idea is to leverage two flexible ML models of the outcome ---one including the treatment variable and one excluding it.
We define the test statistic based on the difference in cross-validation errors between these two models. By randomizing treatments, we can then obtain the null 
distribution of the statistic under certain hypotheses on the treatment effect.
 Any discernible difference between the null distribution and the observed  value of the test statistic is thus evidence of a treatment effect. Crucially, by virtue of randomization, this testing procedure is valid in finite samples, regardless of the choice of ML models, making it a potentially valuable addition to the existing causal ML toolkit. Moreover, our procedures can be powerful thanks to the flexibility of modern ML tools, which we demonstrate in this paper through both theory~(Section \ref{sec:theory}) and a variety of empirical applications~(Section \ref{sec:simu}).

In particular, through a novel power analysis in Section~\ref{sec:res}, we demonstrate that our approach can be more sensitive than methods designed specifically for estimating the average treatment effect, particularly in settings with complex, heterogeneous treatment effects that do not meaningfully contribute to the average treatment effect. From this perspective, the randomization testing framework employed in our work and the causal ML framework can be viewed as complementary and potentially synergistic.

\subsection{Related Work}
In recent years, causal inference has experienced a surge in methods that utilize advanced ML tools. In observational studies, these ML methods are usually employed to estimate average treatment effects or heterogeneous treatment effects; examples include doubly robust estimation~\citep{robins1995semiparametric, bang2005doubly}, the ``X-learner''~\citep{kunzel2019}, double machine learning~\citep{chernozhukov2018double, kennedy2022semiparametric, syrgkanis2019machine}, causal forests \citep{wager2018estimation, oprescu2019orthogonal}, targeted learning~\citep{van2011targeted} and non-parametric Bayes \citep{hill2011bayesian, hahn2020bayesian, green2012modeling, taddy2016nonparametric}. As these methods are primarily designed for non-experimental settings, their validity follows from large-sample properties and rarely have finite-sample guarantees.

In experimental settings, randomization inference provides a robust alternative for causal inference. This approach dates back to the foundational contributions of RA Fisher~\citep{fisher1935design}, who developed what is now known as {\em Fisher's randomization test} (FRT). Unlike asymptotic methods, FRT procedures are finite-sample valid for arbitrary data distributions, as they leverage the known experimental variation in assigning treatments.
These randomization tests are also straightforward to implement, often as permutation tests, making them ideal analytical tools for online A/B experimentation~\citep{kohavi2013online, kohavi2020trustworthy}. 
Thanks to their robustness, FRT procedures have seen renewed interest across many scientific disciplines, 
including causal inference~\citep{imbens2015causal}, experimental economics~\citep{list2019multiple}, and predictive inference, particularly in the form of conformal prediction~\citep{lei2021conformal, tibshirani2019conformal, vovk2005algorithmic}.

Directly relevant to our work are randomization tests for heterogeneous treatment effects~\citep{ding2016variation} and tests for spillover effects under network interference~\citep{aronow2012general, athey2018exact, basse2019randomization, puelz2021graph, basse2024randomization}. See also~\cite{gerber2012field,green2012modeling, grimmer2017estimating} for various applications of randomization inference in field experiments.
However, existing FRT procedures typically rely on linear model specifications, and the use of modern ML tools in this literature remains largely unexplored.

Our paper also contributes to the emerging literature on the use of ML techniques for the design and analysis of experiments~\citep{list2024using, poyarkov2016boosted, imai2025statistical, li2024neyman}. Specifically, we develop randomization tests for experimental studies that address a wide range of applied causal problems, including global treatment effects (Section~\ref{sec:general}), heterogeneous effects (Section~\ref{sec:het}), and spillover effects (Section~\ref{sec:spillover}). Within the same framework, we introduce a new procedure in Section~\ref{sec:imbalance} to detect whether the experiment is imbalanced in certain key covariates and propose ways to address this. The central theme across all these tests is the integration of the finite-sample validity of randomization tests with the predictive power of modern ML tools.

\section{Setup and Main Method}\label{sec:setup}
We have a population of $n$ experimental units indexed by $i$. 
The units are assigned to binary treatments $\ZZ = (Z_1, \dots, Z_n) \in \{0, 1\}^n$ according to a known experimental 
design $\Pn : \{0,1\}^n \to [0,1)$.
We adopt the potential outcomes framework of causal inference~\citep{neyman1923applications, rubin1974estimating}, such that for any population treatment $z\in\{0,1\}^n$, $Y_i(z)\in\mathbb{R}$ denotes the potential outcome of unit $i$ under $z$.
Note that this definition allows the outcome of a unit 
to depend on the treatment of other units, known as interference. 
In settings with no interference, we will use the more common notation $Y_i(1), Y_i(0)$ to denote the treated and control potential outcomes, respectively.

The realized (observed) outcome for unit $i$ is denoted as $Y_i = Y_i(Z)$, and $Y = (Y_1, \ldots, Y_n)$ is the realized population vector.
We further observe covariates $X_i \in \cX \subseteq \R^p$ and use $\XX$ to denote the entire $n \times p$ matrix of covariates. We also use $\ZZ_{-i}$ to denote the sub-vector of $Z$ without the $i$th element corresponding to $Z_i$.

Throughout the paper, we work under the following potential outcome model:
\begin{equation}\label{eq:our_model}
    Y_i(z) = \mu + b(X_i) + z_i h(X_i) + g(\XX, z_{-i}) + \epsilon_i,
\end{equation}
where $b$, $h$, and $g$ are arbitrary functions characterizing the baseline effect, direct effect, and spillover effects, respectively. We assume that these functions are orthogonal to ensure identifiability. We adopt a superpopulation framework with random variables $(\epsilon_i)_{i=1}^n$ corresponding to mean-zero independent noise, such that $\E(\epsilon_i|\XX) = 0$. This type of model specification extends existing models in the causal ML literature~\citep{chernozhukov2018double, kunzel2019} by incorporating  potential spillover effects through function $g$. 
While the assumption of independent errors is not universal  in the causal ML literature, 
the specification in Equation~\eqref{eq:our_model}
allows us to flexibly define a variety of causal questions, including interference, as simple hypotheses on functions $b, h$ or $g$. 
We discuss these questions next. 

\subsection{ML-assisted Randomization Test}\label{sec:general}

We begin with the canonical null hypothesis that the treatment has no effect on the outcomes. This hypothesis is global in the sense that it excludes the presence of either an average treatment effect or any heterogeneous effects. Despite its strength, this hypothesis is usually a good starting point of the empirical investigation,
especially in complex experimental designs where a novel treatment of uncertain efficacy is being tested. 

Under the outcome model of Equation~\eqref{eq:our_model}, we may define the  global null hypothesis and the alternative hypothesis as follows:
\begin{align*}
H_0^{\glob}: h=0, g=0 \quad \text{v.s.} \quad H_1^{\glob}: h \neq 0, g=0.
\end{align*}
We note that the null hypothesis, as defined above, is a statement about the potential outcome function, but it can also be written as a statement on observables through the following conditional independence statement. 
%
\begin{equation*}
    H_0^{glob}: Y_i \indep Z_i \mid \mathbf{X}.
\end{equation*}

The global null hypothesis therefore implies that the treatment is independent of the outcomes and no interference is allowed between units. 
The definition of the alternative hypothesis, on the other hand, sets $g=0$ and so the test aims at detecting only ``direct effects", assuming no interference or spillover effects. We will consider an alternative hypothesis space aimed at detecting potential spillover effects ($g\neq 0$) later in Section~\ref{sec:spillover}.

To test the global null hypothesis $H_0^{\glob}$, we rely on the existence of two ML models on the outcomes, namely
\begin{align*}
    \cM_0^{\glob}: Y_i &\sim X_i, \\
    \cM_1^{\glob}: Y_i &\sim Z_i + X_i.
\end{align*}
With this notation, we mean that $\cM_0$ fits  the outcomes on all covariates, whereas $\cM_1$ also includes the treatment variable in the model.
Next, we define the difference in the cross-validation (CV) errors between the two models as the function
\begin{equation}\label{eq:cv_diff}
    t_n(\YY, \ZZ, \bX) = \CV(\cM_0^{\glob}) - \CV(\cM_1^{\glob}),
\end{equation}
where $\CV(\cM)$ denotes the $k$-fold CV squared loss of model 
$\cM$. 
We are now ready to define our main testing procedure.
\begin{procedure}[ (ML-assisted Randomization Test)]\label{ml-frt}
\begin{enumerate}
    \item Obtain the observed value of the statistic 
    $T_n = t_n(Y, Z, \bX)$ as defined in Equation~\eqref{eq:cv_diff}.
    \vspace{-10pt}
    \item Compute the randomized statistic $t^{(r)} = t_n(Y, Z^{(r)}, \XX)$, $Z^{(r)} \stackrel{iid}{\sim} \Pn$, for $r=1, \dots, R$. 
    \vspace{-10pt}
    
    \item Calculate the one-sided $p$-value: 
    \begin{equation}\label{eq:pval}
    \mathrm{pval}= \frac{1}{1+R} \bigg[ \sum_{r=1}^R \mathbbm{1}\{t^{(r)} > T_n \} + U (1+ m_R)\bigg],
\end{equation}
    where $U\sim\mathrm{Unif}[0,1]$, the standard uniform distribution, and $m_R = \sum_{r=1}^R \mathbbm{1}\{t^{(r)} = T_n \}$ denotes the multiplicity of the randomized statistic values.
\end{enumerate}
\end{procedure}

The following theorem shows that the $p$-value from Procedure~\ref{ml-frt} is valid in finite samples for the null hypothesis of no treatment effect, regardless of the particular choice of models.
\begin{theorem}\label{thm:valid}
Suppose that $H_0^\glob$ holds true. Then, 
\begin{equation*}
    \P\big(\mathrm{pval} \le \alpha) \le \alpha,~\text{for any $\alpha \in [0,1]$ and any $n>0$},
\end{equation*}
where the randomness in $\P$ is with respect to the 
experimental design $\Pn$.
\end{theorem}
\begin{proof}
By virtue of treatment randomization and model~\eqref{eq:our_model}, the null hypothesis $H_0^\glob$ implies that $Y$ and $Z$ are independent conditional on $\XX$, such that $(Y, Z, \XX) \myeq{d}(Y, Z^{(r)}, \XX)$ conditional on $Y,\XX$. The validity of the conditional $p$-value then follows from standard results in the randomization literature~\cite[Theorem 15.2.1]{lehmann2005testing}; \cite[Appendix I]{young2019channeling}. Since the $p$-value is valid conditionally for any $Y, \XX$, it is also valid unconditionally.
\end{proof}
\begin{remark}
The null hypothesis $H_0^\glob$ is known as ``Fisher's sharp null'' hypothesis, and thus Procedure~\ref{ml-frt} is a type of Fisherian randomization test (FRT)~\citep{fisher1935design}. 
However, a distinctive feature of our test lies in the choice of test statistic as the difference in CV errors between two ML models, whereas the original randomization test utilized the classical two-sample $t$-statistic. The power analysis presented in Section~\ref{sec:theory} is therefore a key theoretical contribution of this paper. To our best knowledge, it provides  the first theoretical analysis of ML-assisted randomization tests, as exemplified by Procedure~\ref{ml-frt}.
\end{remark}

\begin{remark}\label{remark2}
Procedure~\ref{ml-frt} is designed to test whether the treatment affects outcomes in any way, but it is not intended to quantify the treatment effect. To determine the directionality of the treatment effect, we propose the residualized procedure described in Section~\ref{sec:res}, which performs a classical randomization test on the residuals from model $\cM_0^{\glob}$. The advantage of Procedure~\ref{ml-frt} lies in its increased power when the objective is solely to test for the existence of a treatment effect, without concern for its directionality. We formally establish this power result in Section~\ref{sec:res}.
 
\end{remark}

\subsection{Related methods}
In this section, we discuss certain aspects of Procedure~\ref{ml-frt}, relating it to methods from experimental design and machine learning, namely ANOVA and variable importance.

First, we note that Procedure~\ref{ml-frt} is akin to ANOVA procedures used in the analysis of experiments. For instance, to test for heterogeneous effects with respect to some covariate $X'$, \citet[Section 9.3.1]{gerber2012field} propose a randomization test  comparing the goodness-of-fit, as measured by the classical $F$-statistic, between a linear model that includes $X'$ and a model that excludes it. 
Procedure~\ref{ml-frt} generalizes such approaches  by incorporating
non-linear ML models and goodness-of-fit measures based on cross-validation errors.

Another related concept is {\em mean decrease accuracy} (MDA), which is widely used to quantify the variable importance in random forest models~\citep{breiman2001random, benard2022mean, ishwaran2024fast}. 
In fact, our test statistic in Equation~\eqref{eq:cv_diff} reduces to 
the original MDA measure in Breiman's seminal work~\citep{breiman2001random} if 
(i) we replace the CV error  with the ``out-of-bag'' prediction error of the random forest model under permutations of $Z$, and (ii) 
$\Pn$ is a completely randomized experiment. 
When $Z_i$ is independent of $X_i$, this measure would be a consistent estimator of $\EX[\VAR(Y | \XX)] / \VAR(Y)$.
This quantity is known as the {\em total Sobol index} of treatment variable $Z$, and measures the contribution of $Z$ in the variance of $Y$, including main effects and interactions.\footnote{In a non-randomized study, the original MDA measure is generally inconsistent for the total Sobol index. See~\cite{benard2022mean} for more details and ways to address this issue.}

In simpler terms, in an A/B experiment where $Z_i$ is completely randomized, the observed value of our test statistic in Equation~\eqref{eq:cv_diff} normalized by the outcome sample variance, 
\begin{equation}
    \frac{ t_n(Y, Z, \XX)} {\hat\sigma_Y^2}, ~\hat\sigma_Y^2 = [1/(n-1)] \sum_{i=1}^n (Y_i - \bar Y)^2~,
\end{equation}
may be interpreted as an estimate of the total Sobol index of the treatment variable. 
An index close to 0 indicates a treatment with small impact on the outcomes, whereas an index with a value near 1 indicates a highly impactful treatment.

Such measures of nonparametric variable importance have also been 
considered in non-causal settings by \cite{williamson2020efficient, williamson2021nonparametric, williamson2023general}.
To highlight some additional key differences, the variable importance measure in these works is a parameter to be estimated, whereas we use the variable importance measure as the test statistic for testing causal hypotheses. Moreover, their methods are based on semi-parametrics and are only valid asymptotically, whereas our method exploits the experimental variation in randomized experiments and is finite-sample valid (per Theorems \ref{thm:valid} and \ref{thm:spillover}).
Our work may thus be considered complementary to these papers.



\subsection{Covariate Balance}\label{sec:imbalance}
Here, we discuss an application of Procedure~\ref{ml-frt} to detect covariate balance. In randomized experiments, covariate balance refers to how similar 
the covariate distributions are between the treated and control groups. Balanced experiments are generally more trustworthy because any observed differences between treated and control units can be confidently attributed to the treatment effect.
Although an A/B experiment is balanced in the limit of infinite data, 
there may be imbalance in finite samples if the experiment is not well designed, 
especially in high-dimensional settings with many covariates.

It should be noted that Procedure~\ref{ml-frt} remains valid even under extreme covariate imbalance as it controls the type I error under all circumstances. However, covariate imbalance may inflate its type II error, leading to non-rejections of the null hypothesis even when the treatment is effective.  
%
To detect covariate imbalance, we can adapt Procedure~\ref{ml-frt} towards estimating the effect of treatment on the covariates. In particular, we can fit two models,
\begin{align*}
    \cM_0^{\im}: X^j &\sim \XX^{-j}, \\
    \cM_1^{\im}: X^j &\sim Z + \XX^{-j},
\end{align*}
where $X^j$ denotes the $j$-th covariate and $\XX^{-j}$ denotes the sub-matrix of $\XX$ without
the $j$-th covariate. In other words, we apply Procedure~\ref{ml-frt} with $X^j$ as the outcome and $\XX^{-j}$ as the covariates. 
Since we know that the treatment has no effect on the covariates by virtue of randomization, 
a rejection in the above test ---appropriately corrected for multiple testing--- would imply that the experiment is imbalanced in covariate $X^j$.

In case the above test rejects, Procedure~\ref{ml-frt} may be adjusted in a way that can increase the power of the randomization test. Following~\cite{hennessy2016conditional}, we propose conditioning the randomization test on the observed covariate imbalance. Specifically, let $\mathrm{pval}_j(Z, \XX)$ denote the $p$-value from the above test of covariate imbalance with respect to $X^j$, 
and let $q$ be its realized value. We then modify Procedure~\ref{ml-frt} by replacing Step 2 with:
\begin{enumerate}
    \item [2'.] Compute the randomized statistic $t^{(r)} = t_n(Y, Z^{(r)}, \XX)$, $Z^{(r)} \stackrel{iid}{\sim} \Pn^j$, for $r=1, \dots, R$,
\end{enumerate}
where $\Pn^j(z) \propto \mathbbm{1}\{ \mathrm{pval}_j(z, \XX)\le q\} \times \Pn(z)$ is the conditional randomization distribution, given that the imbalance $p$-value is less or equal to $q$. The resulting conditional randomization test under Step 2' remains finite-sample valid following an argument similar to Theorem~\ref{thm:valid}. Moreover, this conditional approach is generally more powerful than the corresponding unconditional test in settings with covariate imbalance~\cite[Section 5.2]{hennessy2016conditional}.

\section{Power Analysis}\label{sec:theory}
In this section, we analyze the power properties of Procedure~\ref{ml-frt} against the alternative space
\begin{equation*}
    H_1^{\glob}: h \neq 0, g = 0.
\end{equation*}
The key insight of our theoretical analysis is that higher predictive power from black-box ML models leads to higher power for our randomization tests. This benefit in increased power is expected to be larger  in settings where complex, heterogeneous treatment effects exist but the average treatment effect is relatively weak.
While our primary focus here is the global null hypothesis $H_0^{\glob}$, our simulations in Section~\ref{sec:simu} demonstrate that this insight holds true for more complex treatment effects as well. We begin by stating the main assumptions in our analysis, starting from a simplifying assumption on the experiment design.
\begin{assumption}\label{asmp:iid}
The experiment is a Bernoulli design with probability $\pi$, i.e., 
$
\Pn(z) = \prod_{i=1}^n \pi^{z_i} (1-\pi)^{1-z_i}.
$
The data $(X_i, \epsilon_i)_{i\in[n]}$ are i.i.d. with $\EX(\eps_i | X_i) = 0$ and $\EX(\epsilon_i^2) < \infty$. In addition, $|Y_i|\le M > 0$ with probability one.
\end{assumption}
Clearly, different ML models may affect the performance of Procedure~\ref{ml-frt} differently. To pin down this connection, we define $\cF_1$ as the function class of ML models in $\cM_1$ with domain $\cX\times \{0, 1\}$. Similarly, define $\cF_0$ as the function class of ML models in $\cM_0$ with domain $\cX$. We now introduce a set of regularity conditions for the ML models.

\begin{assumption}\label{asmp:func_class}
The function classes $\cF_1$ and $\cF_0$ satisfy the following properties. 
\begin{enumerate}
\item For some positive constant $M$, classes $\cF_1$ and $\cF_0$ are $M$-uniformly bounded, i.e., 
\begin{equation*}
    \sup_{x\in\cX, z\in\{0, 1\}} f(x, z) \le M,\quad \sup_{x\in\cX} f_0(x) \le M ~\text{for any}~f\in\cF_1, f_0\in\cF_0.
\end{equation*}

\item For a best predictor $f_r^*\in \underset{f\in\cF_1}{\arg\min}~\EX[Y - f(X, Z^{(r)})]^2$, we have $$
\EX[h(X) (f_r^*(X,1) - f_r^*(X,0))] \ge 0.
$$
\item For $\sigma = (\sigma_i)_{i=1}^n$ i.i.d. Rademacher random variables, define
\begin{equation}\label{eq:rademacher_cplx}
    \cR_n(\cF_1; \mathbb{P} ) \coloneqq \frac{1}{n} \underset{(Y, \mathbf{X}, Z)\sim\mathbb{P}}{\E} \underset{\sigma}{\E} \left(\sup_{f\in\cF_1} \left|\sum_{i=1}^n \sigma_i (Y_i - f(X_i, Z_i))^2\right|\right),
\end{equation}
where $\mathbb{P}$ denotes the distribution of observed data $(Y_i, X_i, Z_i)_{i\in[n]}$. Similarly we define $\cR_{n}(\cF_1;\mathbb{P}^{(r)})$ and $\cR_{n}(\cF_0;\mathbb{P})$, where $\mathbb{P}^{(r)}$ denotes the distribution of randomized data $(Y_i, X_i, Z_i^{(r)})_{i\in[n]}$. It holds that $\cR_n(\cF_1; \P), \cR_n(\cF_0;\P), \cR_n(\cF_1;\P^{(r)})$ are $o(1)$.

\end{enumerate}
\end{assumption}
In Assumption~\ref{asmp:func_class}, Condition 1 imposes a 
standard uniform boundedness assumption on the model functions. Condition 3 is related to the Rademacher complexity of function classes. 
This concept is widely used in statistical learning theory \citep{wainwright2019high, boucheron2013concentration,gine2021mathematical}, measures the ``size'' of the given function class and is $o(1)$ for many function classes, such as linear models
\citep{wainwright2019high}, deep neural networks \citep{bartlett2019}, and reproducing kernel Hilbert spaces \citep{bartlett2002rademacher,hur2024reversible}. We further justify Condition 3 in Section~\ref{sec:rademacher} of the Appendix.

Condition 2 in Assumption~\ref{asmp:func_class} is less standard than the other two, and reflects a complex technical constraint. Specifically, it requires that $f_r^*(x,1) - f_r^*(x,0)$, which represents the best predicted treatment effect under the randomized treatment model using $Z^{(r)}$ instead of $Z$, accurately captures on average the direction (sign) of the true treatment effect, $h(x)$.  Given its critical role in the proof, 
 we analyze several examples below that provide insights into this condition.
\begin{example}[Linear model]\label{ex:C2_1}
Suppose that the potential outcome model \eqref{eq:our_model} is linear with 
$\mu=0$, $b(x) = B x$, $h(x)=1$, $g=0$, implying observed outcomes $Y = BX+ Z + \varepsilon$. Moreover, consider linear function classes $\cF_0 = \{x\mapsto bx\}$ and
$\cF_1 = \{(x,z)\mapsto bx + cz\}$, where $b$ and $c$ belong to a large enough bounded region.
To obtain $f_r^*$ we need to solve
\begin{equation}\label{eq:quad}
    \min_{f\in\cF_1} \E[Y - f(X,Z^{(r)})]^2 = \min_{b, c} \E (Y - bX - cZ^{(r)})^2.
\end{equation}
The solution is $b^* = B + \pi(1-c^*) \mu_1 / \mu_2$ 
and $c^* = \frac{\pi(1-r)}{1 - \pi r}$ with $r = \frac{\mu_1^2}{\mu_2} \in [0,1]$, 
where $\mu_1= \E(X)$ and $\mu_2 = \E(X^2)$.
Thus, $f_r^*(x,1) - f_r^*(x,0) = c^*$. Since $h(x)=1$, we obtain
\begin{equation*}
    \EX[h(X) (f_r^*(X, 1) - f_r^*(X,0))] = c^* \ge 0.
\end{equation*}
%
\end{example}
In Example~\ref{ex:C2_1}, the true treatment effect is equal to 1. 
The randomized treatment model, $f(X, Z^{(r)})$,  learns 
a combination of the treatment assignment mechanism and the covariate distribution. 
The estimated treatment effect under this model is biased and equals $c^* \in [0,\pi]$, its exact value depending on the inverse coefficient of variation of $X$ (i.e., parameter $r^{1/2}$). 
However, although the randomized model leads to a biased estimation of the treatment effect, 
it correctly captures the effect's sign, thus satisfying Condition 2.

\begin{example}[Interaction model]\label{ex:C2_2}
Continuing from Example~\ref{ex:C2_1}, suppose that the true outcome 
model is an interaction model such that $h(x)=x$, implying observed outcomes $Y = BX + X Z + \varepsilon$. Consider using the same function classes $\cF_0$ and $\cF_1$ as in Example~\ref{ex:C2_1}. Then, solving Equation~\eqref{eq:quad} yields 
$b^* = B + \pi$ and $c^* = 0$, so that
$$
 \EX[h(X) (f_r^*(X, 1) - f_r^*(X,0))] = \EX(X c^*) = 0.
$$
We see that Condition 2 is satisfied even though $\cF_1$ did not directly take into account the interaction between treatment and covariates, i.e., did not include the term $X Z$ in the model.

\end{example}

\begin{example}[Counter-example for Condition 2]\label{counterex:C2}
In Example~\ref{ex:C2_1}, suppose we define $\cF_1 = \{ (x,z) \mapsto 0\cdot x + c z \}$. Clearly, this is a poor choice of the function class 
because it does not utilize any covariate information.
Suppose also that the mean control potential outcome, $\EX[Y(0)]$ is a large negative number. That is,  $\EX[Y(0)] = B \EX(X) = -x_0$, where $x_0 > 0$ is a large positive constant.
Then, solving Equation~\eqref{eq:quad} yields 
$c^* = B\EX(X) + \pi$, and so
$$
 \EX[h(X) (f_r^*(X, 1) - f_r^*(X,0))] = B\EX(X) + \pi < 0.
$$
We see that Condition 2 is violated under this choice of $\cF_1$. 

\end{example}
Example~\ref{counterex:C2} represents a scenario where the baseline effect is very large relative to the treatment effect. Since $\cF_1$ does not 
use any covariate information, solving Equation~\eqref{eq:quad} leads to a model that is biased towards the dominating baseline effect, such that $c^* <0$ under the randomized model. This example suggests that to guarantee Condition 2 it is necessary to consider large enough function classes that 
on average ``match'' the true effect.

Crucially, if Condition 2 is satisfied, then the variable importance measure $\Delta$ in the theorem below is non-negative, and quantifies the type II error bound for Procedure~\ref{ml-frt}. See Appendix~\ref{sec:proof} for the proof. 

\begin{theorem}\label{thm:power}
Suppose that Assumptions~\ref{asmp:iid} and \ref{asmp:func_class} hold with $M>0$, 
and the number of cross-validation folds is fixed with $k > 1$. Define 
\begin{equation}\label{eq:Delta}
\Delta \coloneqq \inf_{f\in \cF_1} \E [Y - f(X,Z^{(r)})]^2 - 
\inf_{f\in \cF_1} \E [Y - f(X, Z)]^2, 
\end{equation}
where $(Y,X,Z,Z^{(r)})$ is an independent copy of $(Y_i, X_i, Z_i, Z_i^{(r)})$. Then, under the alternative $H_1^{\glob}$, $\Delta \ge 0$, and for some small constant $C>0$,
\begin{align*}
    \P(\mathrm{pval}>\alpha)
    &= O\left( k \exp\left(-\frac{C n \Delta^2}{k M^4}\right) \right).
\end{align*}
\end{theorem}

Theorem~\ref{thm:power} establishes an asymptotic bound for the type II error of our main procedure. 
The proof leverages concentration inequalities from empirical process theory, and is applicable to general machine learning models under the assumptions of the theorem. We note that this type of power result is rare in the randomization literature, 
which typically studies power using CLT-based analyses on classical $t$-statistics~\citep{lehmann2005testing,zhao2021}. 
Technically, the high-level idea behind our proof aligns with \cite{dobriban2022consistency, guo2023invariance} who investigate the minimax optimality of randomization tests in regression settings, but our analysis extends these ideas to a causal inference setting with distinct proof techniques.

From Theorem~\ref{thm:power}, the key quantity in our analysis is
\begin{align*}
    \Delta &= \underbrace{\inf_{f\in \cF_1} \E [Y - f(X, Z^{(r)})]^2}_\text{prediction error with randomized data} - \underbrace{\inf_{f\in \cF_1} \E [Y - f(X, Z)]^2}_\text{prediction error with observed data},
\end{align*}
and acts as the variable importance measure for our test. 
Intuitively, $\Delta$ quantifies the improvement in  prediction accuracy 
when using the real treatment, $Z$, as opposed to a randomized treatment, $Z^{(r)}$, while keeping the observed outcomes, $Y$, fixed.
By Condition 2 of Assumption~\ref{asmp:func_class}, $\Delta \ge 0$, and so 
the real treatment is at least as predictive as its randomized counterpart.
This is reasonable under the alternative hypothesis 
since randomizing the treatment while keeping the outcome fixed effectively removes any dependence between outcomes and treatment.
Moreover, this improvement in prediction captured by $\Delta$  depends not only on the treatment effect but also on the choice of function classes, $\cF_1$ and $\cF_0$. 
In the next section, we compute $\Delta$ for specific function classes to illustrate these points. 

\subsection{Comparison to a Residualized Method}\label{sec:res}

In this section, we compare the power of Procedure~\ref{ml-frt} (referred to as ML-FRT) to the residualized method (RES) outlined in Remark~\ref{remark2}. 
This method is an alternative to Procedure~\ref{ml-frt} and relies on the idea of covariate adjustment~\citep{rosenbaum2002covariance,zhao2021}. 

In particular, in the residualized method, one first fits a model of $Y$ on covariates $\XX$ ---e.g., 
through linear models \citep{tukey1993tightening, rosenbaum2002covariance}, generalized linear models \citep{gail1988tests}, or nonparametric regression~\citep{raz1990testing}--- to obtain the residuals, and then performs the classical FRT on the residuals. 

Using our notation, the procedure of the residualized method can be formalized as follows.
\begin{enumerate}
    \item Utilize a machine learning model to fit $Y_i \sim X_i$.  Obtain a fitted model $\widehat{m}(X_i)$ and residuals $\widehat{\eps}_i = Y_i -  \widehat{m}(X_i)$.
    \item Run an FRT with the test statistic
    \begin{equation*}
       t^{\res}(Y, Z, \bX; \widehat{m}) =  \CV(\widehat{\eps}_i \sim 1) - \CV(\widehat{\eps}_i \sim 1 + Z_i).
    \end{equation*}
\end{enumerate}
In step 2, we define the statistic using CV errors so that we can directly invoke Theorem~\ref{thm:power}  for our  power analysis. 
One can expect this CV-based test statistic to perform similar to the difference-in-means estimator, since the fitted coefficient of $Z_i$ in $\widehat{\eps}_i \sim 1 + Z_i$ is exactly the difference-in-means.
In step 1 of the residualized method, the best predictor that minimizes $\E [Y_i - m(X_i)]^2$, under Model~\eqref{eq:our_model} and Assumption~\ref{asmp:iid}, is 
\begin{equation*}
    m^*(x) = b + \mu(x) + \pi h(x).
\end{equation*}
We will assume that the fitted model $\widehat{m}$ is close enough to $m^*$ in the following sense.
\begin{assumption}\label{asmp:consistent}
As $n\to\infty$ the fitted model $\widehat{m}(x)$ satisfies 
\begin{equation*}
    \sup_{x\in\cX}|\widehat{m}(x) - m^*(x)| \to 0.
\end{equation*}
\end{assumption}
To apply our power analysis, we also impose the following additional assumptions.
\begin{assumption}\label{asmp:iid_res}
There exists a constant $M>0$, such that $|h(X_i)|\le M/2$ and $|\eps_i|\le M/2$ with probability one. 
\end{assumption}
Assumption~\ref{asmp:iid_res} implies that $|Y_i - m^*(X_i)| = |(Z_i-\pi) h(X_i) + \eps_i |\le M$, and so the residuals remain bounded. Under this assumption, we can explicitly construct the function classes so that step 2 of the residual method falls under our power analysis. 
That is, under these assumptions, we may express the residual method as a case of 
Procedure~\ref{ml-frt} where $\cF_0$ contains the constant functions with a range between $[-M,M]$ and $\cF_1$ contains the linear function class with coefficients in $[-M/2, M/2]$.

To aid with interpretation, we impose one final assumption on Procedure~\ref{ml-frt} so that the function class $\cF_1$ contains the best predictors under Model~\eqref{eq:our_model}.
\begin{assumption}\label{asmp:func_class_inf}
We have $\mu + b(x) + \pi h(x), \mu + b(x) + z h(x)\in\cF_1$.
\end{assumption}
The following proposition characterizes the performance of the residual approach relative to Procedure~\ref{ml-frt} in terms of parameter $\Delta$. See Appendix~\ref{sec:proof} for the proof.
\begin{proposition}\label{prop:res}
Suppose Assumptions~\ref{asmp:iid}-\ref{asmp:func_class_inf} hold with the same constant $M>0$ and that $k>1$ is a fixed constant. Let $(Y, X, Z)$ be an independent copy of $(Y_i, X_i, Z_i)$. Then the type II error bound in Theorem~\ref{thm:power} holds for both the residualized method (RES) and Procedure~\ref{ml-frt} (ML-FRT) with the following $\Delta$ terms:
    \begin{align*}
        \Delta^{\res} &= \pi(1-\pi) [\E h(X)]^2,\\
        \Delta^{\ml} &= \pi(1-\pi) \E[h^2(X)].
    \end{align*}
\end{proposition}
Proposition~\ref{prop:res} demonstrates that the variable importance measure $\Delta$ of the residual method is always dominated by that of Procedure~\ref{ml-frt}. More specifically, the residual method detects the treatment effect only through the average treatment effect (ATE) quantity, $\E h(X)$. 
This ``limitation'' of the residual method is not surprising, since it has been shown to be asymptotically valid for testing that ATE equals zero \citep{zhao2021}.
In contrast, Procedure~\ref{ml-frt} captures both the ATE and heterogeneous treatment effects through $\E[h^2(X)]$. Since $\Delta^{\ml}\ge \Delta^{\res}$, the main 
result of Proposition~\ref{prop:res} is that Procedure~\ref{ml-frt} enjoys a better bound on its Type II error than simple covariate adjustment through residualization.

However, Proposition~\ref{prop:res} does not necessarily imply 
that Procedure~\ref{ml-frt} is more efficient compared to the residual method. To obtain a stronger result on relative efficiency, we will work with an idealized version of these procedures as follows. 
For Procedure~\ref{ml-frt}, we define a deterministic test $\phi^{\ml} = \mathbb{I}\{t_n(Y, Z, \bX) > q^{\ml}_{n, \alpha}\}$, where $q^{\ml}_{n, \alpha}$ is the exact $1-\alpha$ quantile for $t_n(Y, Z, \bX)$ under $H_0^{\glob}$. Similarly, define $\phi^{\res} = \mathbb{I}\{t^{\res}(Y, Z, \bX; \widehat{m}) > q^{\res}_{n, \alpha}\}$ for the residualized method.

We assume the following property of the test statistics in $\phi^{\ml}$ and $\phi^{\res}$.
\begin{assumption}\label{asmp:large_deviation}
There exists a rate function $I(\cdot)$ that is continuous and strictly increasing on $(0, \infty)$, such that for any $x>0$,
\begin{align}
    &\lim_{n\to\infty} \frac{1}{n} \log \P(t_n(Y, Z, \bX) - \Delta^{\ml} < -x) \ge - I(x),\label{eq:large_dev}\\
    &\lim_{n\to\infty} \frac{1}{n} \log \P(t^{\res}(Y, Z, \bX; \widehat{m}) - \Delta^{\res} < -x) \le - I(x).\nonumber
\end{align}
\end{assumption}
The role of Assumption~\ref{asmp:large_deviation} is 
to provide a lower bound on the tails of the test statistics' sampling distributions. 
This strengthens the conditions underlying Theorem~\ref{thm:power}, 
which provide an upper bound for the sampling distribution since
\begin{gather}
    \P(t_n(Y, Z, \bX) - \Delta^{\ml} < -x ) \le \P(|t_n(Y, Z, \bX) - \Delta^{\ml}| > x) 
    = O\left(k \exp\left(-\frac{cnx^2}{kM^4}\right)\right), \nonumber\\
    \Rightarrow \underset{n\to\infty}{\lim\sup} \frac{1}{n}\log\P(t_n(Y, Z, \bX) - \Delta^{\ml} < -x ) \le - \frac{cx^2}{kM^4} \eqqcolon - I^{\ml}(x),\label{eq:tail_upper}
\end{gather}
where $c$ is a constant obtained in the proof. We emphasize that Assumption~\ref{asmp:large_deviation} does not 
require that  $t_n(Y,Z,\bX)$ or $t^{\res}(Y,Z,\bX)$ converge in probability to constants $\Delta^{\ml}$ and $\Delta^{\res}$, respectively. Moreover, it does not require 
that  $t_n(Y,Z,\bX)-\Delta^{\ml}$ and $t^{\res}(Y,Z,\bX) - \Delta^{\res}$ converge in distribution, up to some rescaling. In this sense, Assumption~\ref{asmp:large_deviation} is weaker than asymptotic normality or consistency. 
However, verifying Assumption~\ref{asmp:large_deviation} is challenging because it may require
strong conditions on the particular model class for the CV error as well as the joint data distribution. We consider this as an interesting direction for future work.

The following result shows the relative efficiency between $\phi^{\ml}$ and $\phi^{\res}$. 
\begin{proposition}\label{prop:relative_efficiency}
Suppose that Assumptions \ref{asmp:iid}-\ref{asmp:iid_res}, \ref{asmp:large_deviation} hold. 
Moreover, suppose $\mu + b(x)\in\cF_0$ and there exists $f\in\cF_1$ such that 
$\mu + b(x) =f(x,1)=f(x,0)$ for all $x\in\cX$.  Under the alternative hypothesis $H_1^{\glob}$ with $\E h(X) \neq 0$, we have
\begin{align*}
\underset{n\to\infty}{\lim\inf} \frac{1}{n} \log \frac{\P(\phi^{\ml} = 0)}{\P(\phi^{\res} = 0)} \ge 
I(\Delta^{\res}) - I(\Delta^{\ml}) \le 0~.
\end{align*}
\end{proposition}
Intuitively, Proposition~\ref{prop:relative_efficiency} suggests that, under the alternative, we have:
\begin{equation}\label{eq:relative_eff}
    \frac{\text{Type II error of ML-FRT}}{\text{Type II error of RES}} \ge 
    \exp\left\{- n [I(\Delta^{\ml}) - I(\Delta^{\res})]\right\}.
\end{equation}
To make a concrete comparison, consider the result of Theorem~\ref{thm:power} 
where $\Delta^{\ml} - \Delta^{\res} = \pi(1-\pi)\VAR(h(X)) > 0$. Then, 
$I(\Delta^{\ml})> I(\Delta^{\res})$, and the right-hand side in Equation~\eqref{eq:relative_eff} converges to zero at an exponential rate. In this case, the type II error of Procedure~\ref{ml-frt} is exponentially smaller than that of the residualized method, and their difference is captured by $I(\Delta^{\ml}) - I(\Delta^{\res})$. 

\subsection{Sample Size Determination}\label{sec:sample_size}
In experimental design, sample size determination helps to ensure accuracy of the treatment effect estimate~\citep[Section 8.1.2]{cox2000theory}. 
A typical choice is $n \propto \sigma^2/d^2$, where $d$ is the width of confidence intervals at a desired confidence level, and $\sigma^2$ is an estimate of residual variation. Other approaches have been developed to incorporate covariate adjustment~\citep{turner2012covariate, Schuler2020Increasing} and estimators based on semi-parametric efficiency theory \citep{schuler2021designing}.


Our result on power (Theorem \ref{thm:power}) can be adapted to address this practical concern.
Suppose one wants to test $H_0^{\glob}$ against $H_1^{\glob}$ using Procedure~\ref{ml-frt}, and the test is required to achieve at least $80\%$ power if there is a nonzero treatment effect. Using the proof of Theorem~\ref{thm:power} (i.e., Lemma~\ref{lem:power}), we obtain that at any sample size $n>0$, 
\begin{equation}\label{eq:sample_size}
    \P(\mathrm{pval}>\alpha) \le 4R\left(2k \exp\left(-\frac{n L^2}{32 k M_0^2}\right) +  \exp\left(-\frac{(k-1) n L^2}{128 k M_0^2}\right) \right). 
\end{equation}
In this expression, $L$ captures $\Delta$ up to a certain Rademacher complexity, and $M_0$ is a boundedness constant for the loss function, which are defined in the proof. Then, by setting the right hand side of Equation~\eqref{eq:sample_size} equal to $0.2$, we solve for variable $n$ to obtain the required sample size.

In practice, $L$ and $M_0$ are unknown given a real data set, but we may be able to estimate them using auxiliary data sets. Suppose, for instance, that a pilot experiment is available and provides auxiliary data that has the same distribution as the real data. Then, we can estimate $L$ and $M_0$ using the auxiliary data as follows:
$$
\widehat{L} = \CV(\mathcal{M}_0^{\glob}) - \CV(\mathcal{M}_1^{\glob}),\quad \widehat{M}_0 = \max\{\CV(\mathcal{M}_0^{\glob}), \CV(\mathcal{M}_1^{\glob})\}.
$$
If auxiliary data are not available, we could resort, as a reasonable heuristic, to applying the above formulas in the available sample. Moreover, we may omit the constant $R$ in practice, because it comes from a worst-case analysis on the type II error that tends to be overly conservative. We demonstrate a sample size calculation based on the methods of this section in Appendix \ref{sec:sample_size_example}.

\section{Treatment Effect Heterogeneity}\label{sec:het}
Next, we consider the null hypothesis that the treatment effect is homogeneous across units. This question is important in various contexts, including cost-benefit analyses for public policy and business applications, as it determines whether targeting treatments to specific subpopulations can improve overall outcomes.

Under the outcome model \eqref{eq:our_model}, we may define the  null hypothesis as follows:
\begin{align*}
    H_0^{\het}: h = \text{constant}, g=0, \quad \text{v.s.} \quad H_1^{\het}: h \neq \text{constant}, g=0.
\end{align*}
Under model~\eqref{eq:our_model}, the null hypothesis $H_0^{\het}$ implies the observable model $Y_i = b(X_i) + \tau Z_i + \eps_i$, for some fixed but unknown treatment effect parameter $\tau \in\mathbb{R}$. Equivalently, 
we may express the null hypothesis as the following conditional independent statement:
\begin{equation*}
    H_0^{\het}: Y_i - \tau Z_i \indep Z_i \mid \mathbf{X},~\text{ for some } \tau \in \R.
\end{equation*}

To test $H_0^{\het}$, we may follow a similar approach to \cite{ding2016variation}, 
and majorize the $p$-value in Procedure~\ref{ml-frt} by enumerating plausible values for $\tau$.  Specifically, we first compute $Y_i^{0} = Y_i - \tau_0 Z_i$ for some fixed $\tau_0\in\R$ and then apply Procedure~\ref{ml-frt} on data $\{(\YY_i^0,Z_i,X_i)\}_{i=1}^n$  to obtain a $p$-value, $\mathrm{pval}(\tau_0)$. Finally, we maximize the $p$-values over a grid of values for  $\tau_0$:
\begin{equation}\label{eq:pval_het}
    \mathrm{pval}^{\mathrm{het}}_n = \sup_{\tau_0\in \R} \mathrm{pval}(\tau_0).
\end{equation}
Note that $\mathrm{pval}(\tau)$ is valid in finite samples under the null hypothesis. Due to majorization, this implies that $\mathrm{pval}^{\mathrm{het}}_n$ is finite-sample valid as well. 
Moreover, the computation in Equation~\eqref{eq:pval_het} can be completely parallelized over the available grid.
In case a $(1-\gamma)$-level confidence interval, say $\mathrm{CI}_\gamma$, is available for $\tau$, we may simplify computation further by using the method of \cite{berger1994pvalues} to define:
$ \mathrm{pval}^{\mathrm{het},\gamma}_n =\sup _{\tau_0 \in \mathrm{CI}_\gamma} \mathrm{pval}\left(\tau_0\right)+\gamma$. 
Under this construction, $\mathrm{pval}^{\mathrm{het},\gamma}_n$ remains a valid $p$-value due to the correction term $\gamma$. See Appendix \ref{het_validity} for proofs on the results of this section.
 It may be  concerning that the above majorization of the $p$-value might lead to a loss of power in the randomization test. However, in Section~\ref{sec:simu_het}, we empirically demonstrate that our heterogeneity test maintains competitive power against popular alternatives.

\section{Spillover Effects due to Network Interference}\label{sec:spillover}

Spillover effects occur when non-treated units are indirectly influenced from the treatments on other units. Such effects are pervasive in online social networks and similar interconnected systems, where participants interact or share information. In such settings,  the estimation of treatment effects may be biased, either underestimating or overestimating the treatment's actual impact~\citep{sobel2006randomized, toulis2013estimation}.

Under the potential outcome model \eqref{eq:our_model}, we may test for spillovers using the hypotheses:
\begin{equation*}
    H_0^{\spill}: h \neq 0, g =0, \quad \text{v.s.} \quad H_1^{\spill}: h \neq 0, g \neq 0.
\end{equation*} 
This hypothesis can be also be expressed as:
\begin{equation*}
    H_0^{\spill}:  Y_i \indep Z_{-i} \mid Z_i, \mathbf{X}.
\end{equation*}
Note that a direct application of Procedure \ref{ml-frt} is invalid for $H_0^\spill$ because, unlike the sharp global null, the potential outcomes of a unit may change under randomized treatments that change the unit's individual treatment status.
Indeed, hypothesis $H_0^\spill$ only implies that the outcomes of a unit $i$ remain unchanged if the individual treatment is fixed.

To test $H_0^\spill$, we can leverage a recent line of work in 
conditional randomization tests under interference~\citep{aronow2012general, basse2019randomization, basse2024randomization}. 
The key methodology in these papers proceeds in two steps. 
First, we assume that the spillover function $g$ depends in some unknown way to an observable social network ${\bf A} \in \{0,1\}^{n\times n}$ between units. For example, $g(\XX, Z_{-i}) =  \sum_{j\in [n]} A_{ij} Z_j := A_{i.}^\top Z$ implies that unit outcomes depend on the number of 
the unit's treated neighbors in ${\bf A}$ under $Z$.
Second, we choose a random set of units $\cI\subset[n]$, known as ``focal units'', and then perform a conditional randomization test given this selection.
Roughly speaking, this conditional test aims to associate the ``network treatment'' of units in $\cI$ with their outcomes, while keeping their individual treatments, $Z_{\cI}$, fixed. 

Concretely, given $\cI$ (e.g., a random half of the population), we fit the following models:
\begin{align*}
    \cM_0^{\spill}: Y_i &\sim Z_i + X_i,~i\in\cI\\
    \cM_1^{\spill}: Y_i &\sim Z_i + A_{i.}^\top Z + \XX,~i\in\cI.
\end{align*}
This leads to the test statistic
\begin{equation}\label{eq:spillover_stat}
    t_n(Y,Z,\XX; {\bf A}) = \CV(\cM_0^{\spill}) - \CV(\cM_1^{\spill}).
\end{equation}
The idea is that, if spillover effects exist ($g\neq 0$), then these will be partially captured by $\cM_1$ through the terms ${\bf A}_{i.}^\top Z$ and ${\bf X}$, but they will not be captured by $\cM_0$, leading to rejection of the null hypothesis. To test $H_0^\spill$ we modify Procedure 1 by replacing Steps 1 and 2 with:
{\em 
\begin{enumerate}
 \item [1'.]  Obtain the observed value of the statistic 
    $T_n = t_n(Y, Z, \XX; {\bf A})$ as defined in~
    \eqref{eq:spillover_stat}.
    \item [2'.] Compute the randomized statistic $t^{(r)} = 
    t_n(Y, Z^{(r)}, \XX; {\bf A})$, $Z^{(r)} \stackrel{iid}{\sim} \Pn^{\cI}$, for $r=1, \dots, R$,
\end{enumerate}
}
\noindent where $\Pn^{\cI}(z) \propto \mathbbm{1}\{ z_{\cI} = Z_{\cI} \} \times \Pn(z)$ is the 
conditional randomization distribution given that the treatments of focal units in $\cI$ stay fixed to their realized values, $Z_{\cI}$. For instance, in a simple A/B experiment, $\Pn^{\cI}(z)$  amounts to permuting the treatments of non-focal units. 

We emphasize that the validity of our proposed method does not require that the term $A_{i.}^\top Z$ captures correctly the spillover effect in its entirety. Other choices for the treatment spillover term could be used. If full flexibility is desirable, 
we could even use the entire matrix $\bf A$ as a regressor in the full model:
\begin{align*}
    \cM_1^{\spill}: Y_i &\sim Z + {\bf A} + \XX,~i\in\cI.
\end{align*}
See Section~\ref{sec:application_interference} for a practical example on other possible choices regarding the term for the spillover treatment effect.
The following theorem establishes the validity of the modified randomization 
procedure outlined above.

\begin{theorem}\label{thm:spillover}
Suppose that $H_0^\spill$ holds true. Then,
\begin{equation*}
    \P\big(\mathrm{pval} \le \alpha \mid \cI) \le \alpha,~\text{for any $\alpha \in [0,1]$ and any $n>0$},
\end{equation*}
where the randomness in $\P$ is with respect to the 
conditional randomization distribution $\Pn^{\cI}$.
\end{theorem}

\begin{remark}
    We note that the validity of our procedure for the spillover effect does not depend on using 
 a cross-validation method that is valid for network data. This is because, 
    under the null hypothesis, $H_0^{\spill}$, there is no interference, and the validity of the 
     $p$-value from our procedure relies solely on using the correct conditional randomization distribution, $\Pn^{\cI}$. This reflects an important advantage of the Fisherian randomization framework: it allows the use of complex test statistics ---such as those based on cross-validation--- without requiring assumptions about the asymptotic properties of these statistics.
\end{remark}

\section{Simulations}\label{sec:simu}
In this section, we numerically validate Theorem~\ref{thm:power}, and apply the proposed randomization tests to study heterogeneous treatment effects and spillover effects.

\subsection{Numerical Validation of Theorem~\ref{thm:power}}\label{sec:num_validation}
Here, we demonstrate Theorem~\ref{thm:power} by studying how the choice of the function class affects the power of Procedure~\ref{ml-frt} through parameter $\Delta$. We consider three function classes of decreasing size: random forests, linear models with interaction, and linear models without interaction. We use $
    \widehat{\Delta} =\CV(\mathcal{M}_0^{\glob}) - \CV(\mathcal{M}_1^{\glob})$ to approximate $\Delta$. Based on the proof of Theorem~\ref{thm:power}, $\widehat{\Delta}$ converges in probability to $\Delta$ in regular setups. 

We set $n = 100$, $p = 2$, and $X_i \stackrel{iid}{\sim} \calN(0, \Sigma)$, where $\Sigma$ is a randomly generated correlation matrix based on the R package \texttt{randcorr}. In this package, one first represents the correlation matrix using Cholesky factorization and hyperspherical coordinates (i.e., angles) \citep{pourahmadi2015randcorr}, and then samples the angles from a particular distribution \citep{makalic2022efficient}. We specify model \eqref{eq:our_model} by setting $b(x) = 0.1 x^\top\beta$, random coefficients $\beta\sim\mathrm{U}([1, 5]^p)$, $g = 0$, 
and
\begin{equation*}
h(x) = \begin{cases}
    \tau + \tau \min \{\frac{2}{x_1}, 10\} &~\text{if}~x_1>0 \\
    \tau + \tau \max \{\frac{2}{x_1}, -10\} &~\text{if}~x_1< 0 \\
    \tau  &~\text{otherwise},
\end{cases},
\end{equation*}
with $\tau \in \{0, 0.1, \dots, 0.5, 1, 1.5, 2\}$, and $\epsilon_i \sim\calN(0, 0.1^2)$. The experimental design considered here is an i.i.d. Bernoulli design with treatment probability 0.5.

Figure~\ref{fig:power_SNR} shows the power (rejection rates) and $\widehat{\Delta}$ for different function classes over $\tau$, based on $R = 100$ and 1,000 independent replications of the data generating process. Observe that the power increases as $\widehat{\Delta}$ increases, which is aligned with the type II error bound derived in Theorem~\ref{thm:power}. Moreover, under each alternative hypothesis, the richest function class (i.e., random forests, ``${\mathrm{RF}}$"), achieves the largest $\widehat{\Delta}$ and and also the highest power. These results highlight the theoretical insight that better prediction through ML model leads to higher power for our randomization tests.



\begin{figure}[t]
    \centering
    \includegraphics[width=.8\linewidth]{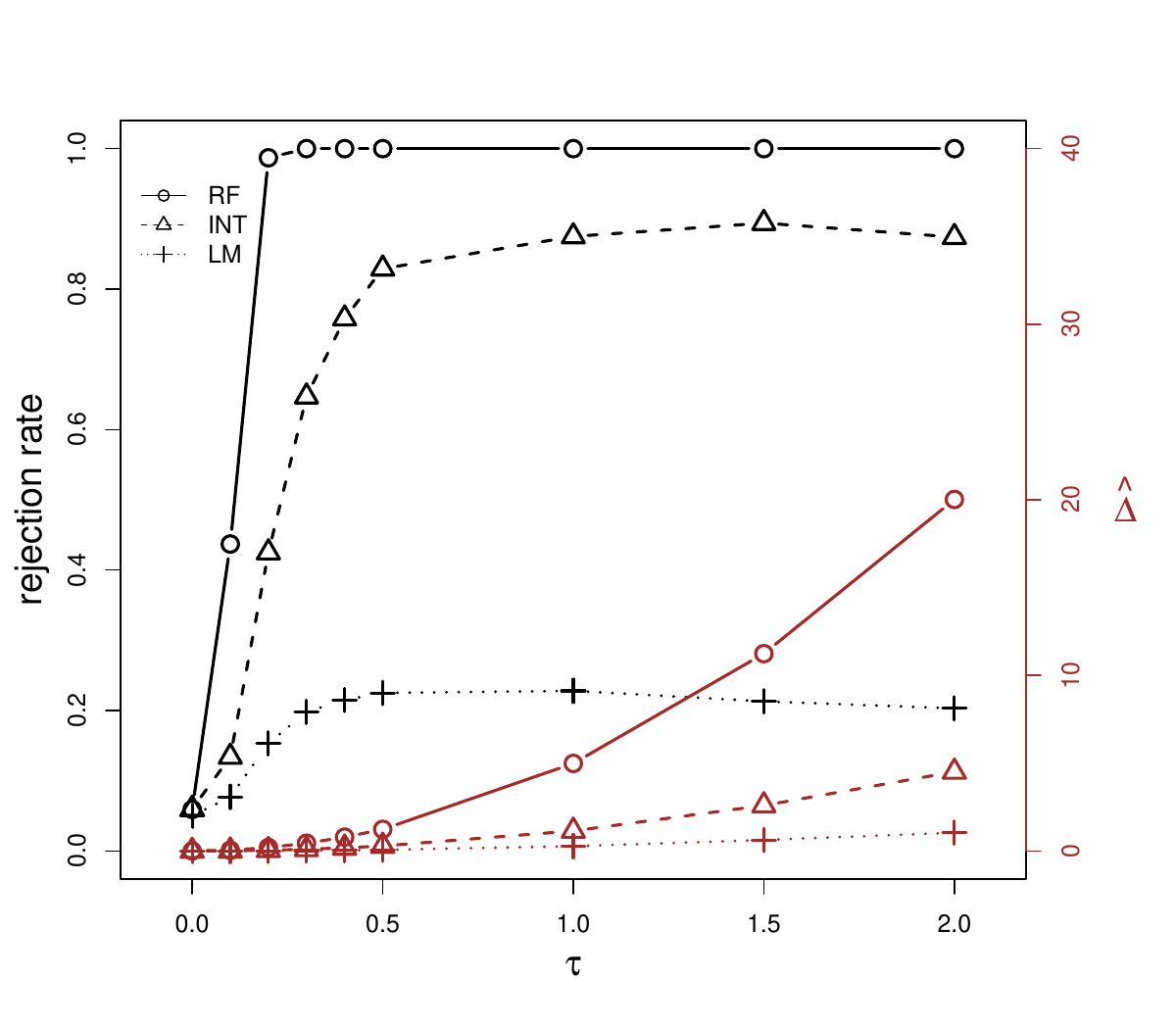}
    \caption{Rejection rates (in black) and signal strength $\widehat{\Delta}$ (in brown) across $\tau$. In the figure, $\mathrm{RF}$, $\mathrm{INT}$, and $\mathrm{LM}$ stand for the function classes of random forest, linear model with interactions, and linear model without interactions.}
    \label{fig:power_SNR}
\end{figure}

\subsection{Testing for Heterogeneous Treatment Effects}
\label{sec:simu_het}
Here we study the randomization test of Section~\ref{sec:het} for treatment effect heterogeneity by adapting an experimental setup from~\cite{kunzel2019}. The setup is an i.i.d. Bernoulli design with treatment probability 0.5. We set $n = 100$, $p = 5$, and $X_i \stackrel{iid}{\sim} \calN(0, \Sigma)$, where $\Sigma$ is a sampled as in Section~\ref{sec:num_validation}. We consider two main setups:
\begin{enumerate}[(A)]
    \item Model with a linear heterogeneous treatment effect:
    $$
    y_i(z) = -0.05 X_i^\top \beta_0 + 0.5 \tau^H z_i \times X_i^\top \beta_1 + \epsilon_i,
    $$
    where $\epsilon_i\sim\calN(0, 1)$ and $\beta_0, \beta_1\stackrel{iid}{\sim} \mathrm{U}([1,30]^p)$. We consider $\tau^H \in \{0, 0.1, \dots, 1\}$ as different values for the heterogeneity parameter ($\tau^H=0$ means no heterogeneity).
   
    \item Model with a nonlinear heterogeneous treatment effect:
    \begin{equation*}
       y_i(z) = \mathbbm{1}(X_{i1}<0.5) - 1.5 \mathbbm{1}(X_{i2}>-0.5) + 
       \tau^H z_i\times \{ 2\mathbbm{1}(X_{i1}<0.5) - 3 \mathbbm{1}(X_{i2}>-0.5)\}+ \epsilon_i, 
    \end{equation*}
 where $\epsilon_i \stackrel{iid}{\sim}\calN(0, 1)$ as before.
    We consider $\tau^H\in\{0, 0.5, \dots, 3\}$.
\end{enumerate}

Through these definitions, Setup (A) represents a setting with simple, linear heterogeneous treatment effect. Setup (B) represents a setting with a nonlinear heterogeneous treatment effect that is challenging  to detect.

We implement our tests using random forests (named $\mlfrt$) in the test statistic~\eqref{eq:cv_diff} and linear models with an interaction term $Z_i X_i$ (named $\lmfrt$). That is, we fit $Y_i\sim Z_i + X_i$ in $\mathcal{M}_0^\glob$ and $Y_i\sim Z_i \times X_i$ in $\mathcal{M}_1^\glob$. We compare our approach to the tests proposed in \cite{ding2016variation}, which also provide finite-sample valid $p$-values for testing heterogeneity, albeit without the use of ML models. As the test statistic, \cite{ding2016variation} use the variance ratio ($\texttt{VR}$) that computes the ratio of variances from treatment and control groups, and the shifted KS statistic ($\texttt{SKS}$) that computes the KS distance between the treated outcomes and controlled outcomes shifted by a given constant treatment effect. 

Figure~\ref{fig:het_eff} shows the rejection rates under the two setups described above, based on $R = 1,000$ randomizations and 100 independent replications of the data generating process. First, we observe that under the null hypothesis of no heterogeneity ($\tau^H=0$), all methods have a correct type I error control. 
Under the alternative ($\tau^H\neq 0$), $\lmfrt$ achieves similar power in the linear Setup (A) and is trailed by $\mlfrt$.
In the nonlinear Setup (B), $\mlfrt$ achieves the highest power among all randomization-based tests. The intuition is that flexible random forests can better capture the nonlinear heterogeneous treatment effect, leading to higher power per Theorem~\ref{thm:power}. 
As an aside, we note that the $\texttt{VR}$ statistic becomes powerless under the alternative $\tau^H = 1$ in the nonlinear setup because $\texttt{VR}$ only leverages the variances within treated-control groups, and thus fails to detect heterogeneous effects that do not affect the variance.


\begin{figure}[t!]
    \centering
    \subfloat[Linear]{{\includegraphics[width=.51\linewidth]{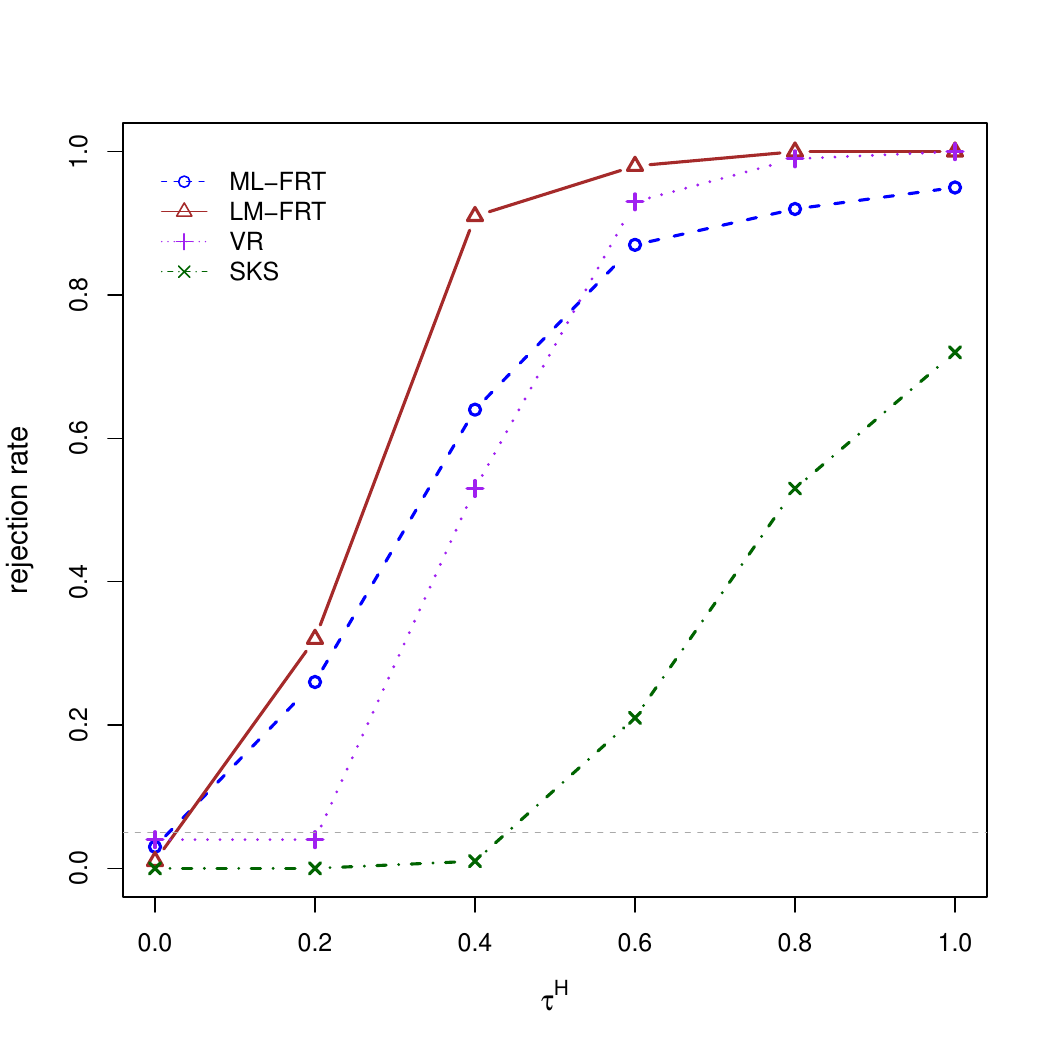}}}%
    \subfloat[Nonlinear]{{\includegraphics[width=.51\linewidth]{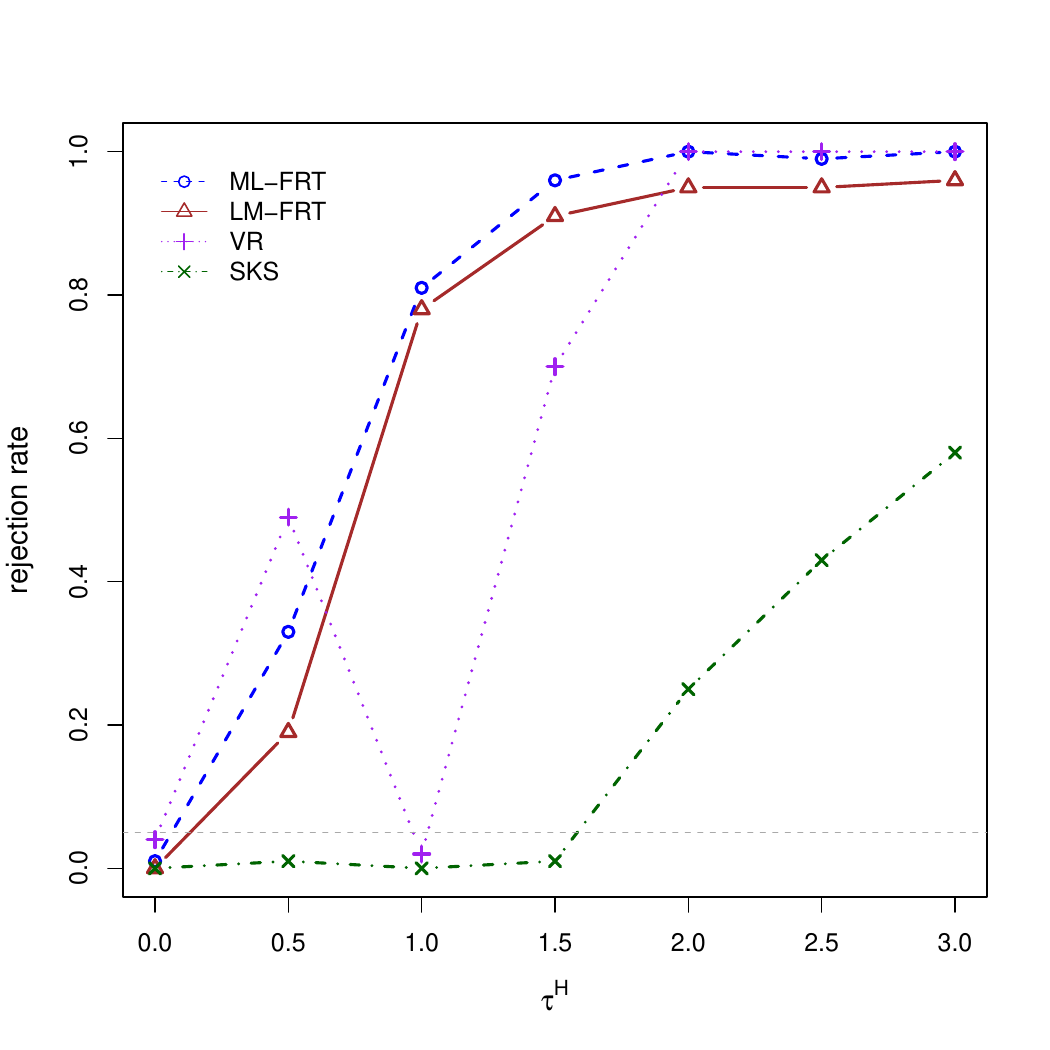}}}%
    \caption{Rejection rates for heterogeneous treatment effect across $\tau^H$.}
    \label{fig:het_eff}
\end{figure}

\subsection{Testing for Spillover Effects}\label{sec:application_interference}
Here we study the randomization test of Section~\ref{sec:spillover} for testing interference by considering a clustered interference setup following \cite[Section 6.2]{puelz2021graph} and \cite[Section 8.1]{basse2018analyzing}. We consider $n = 300$ with 20 randomly generated clusters, where only units within a cluster are connected to each other. 

This results in a block-structured adjacency matrix ${\bf A}$. The experimental design considered here is a type of two-stage randomization that is common when clusters are small (e.g., households). 
The treatment assignment according to this design proceeds in the following two steps. 
First, we select a random half of the clusters and, second, we randomly assign one unit from each selected cluster to get treated. 
Let $R_i(z) = \sum_{j\neq i} A_{ij} z_j \in \{0,1\}$ denote the ``neighborhood treatment'' of $i$. Each unit is assumed to have three possible treatment statuses: 
treated ($Z_i = 1$), spillover ($Z_i=0$ and $R_i(Z) = 1$), 
and pure control ($Z_i = 0$ and $R_i(Z) = 0$).  

We consider the following two setups, 
setting $p=2$ and $X_i\stackrel{iid}{\sim}\calN(0, I_2)$ in both of them.
\begin{enumerate}[(A)]
    \item Model with a constant spillover effect:
    \begin{equation*}
        y_i(z) = 2 + 1.5 z_i + \tau^S (1-z_i)\times R_i(z) + \epsilon_i,
    \end{equation*}
    where $\eps_i$ follows a clustered structure: $\epsilon_i = e_{c(i), 0} + Z_i e_{c(i), 1} + e_{i}$, $e_{c(i), 0}, e_{c(i), 1}\stackrel{iid}{\sim}\calN(0, 0.1^2)$ and $e_i\sim\calN(0, 0.5^2)$. Here, $c(i)\in\{1, \dots, 20\} $ denotes $i$'s cluster. 
    In this setup, the covariates $X$ are irrelevant.
    \item  Model with a nonlinear spillover effect:
    \begin{equation*}
         y_i(z) = 2 + 1.5 z_i + 0.5 \tau^S (1-z_i)\times R_i(z)\times \left(3\Ind\{X_{i2}>-0.5\} - 2\Ind\{X_{i1}<0.5\} \right) + \epsilon_i,   
    \end{equation*}
    where $\eps_i$ follows a clustered, heteroskedastic structure: $\epsilon_i = e_{c(i), 0} + z_i e_{c(i), 1} + (1-z_i)e_{i,0} + Z_i e_{i,1}$. Here, $e_{c(i), 0}, e_{c(i), 1}\stackrel{iid}{\sim}\calN(0, 0.1^2)$, $e_{i, 0}\sim\calN(0, X_{i1}^2/3^2)$, and $e_{i, 1}\sim \normal(Y_{i,2}, 0.5^2X_{i2}^2)$.
\end{enumerate}

We vary $\tau^S$ across $\{0, 0.1, \dots, 1\}$ and $\{0, 0.2, \dots, 2\}$ under Setups (A) and (B), respectively. We choose the test statistic induced by $\cM_0^{\spill}: Y_i \sim Z_i + X_i$ and $\cM_1^{\spill}: Y_i \sim Z_i + \sum_{j\neq i} {A}_{ij}Z_j + X_i$, $i \in \cI$. As noted in Section \ref{sec:spillover}, this gives one possible candidate for detecting spillovers among many. 
We consider another candidate based on the edge-level-contrast statistic (\texttt{ELC}) \citep{athey2018exact}, which can be interpreted as the difference in the average outcomes between two groups of focal units: those whose neighbors have been exposed to the treatment and those whose neighbors have not been exposed. We again compare two variants of our approach, one with random forests (\mlfrt) and the other with linear models (\lmfrt). 
The focal units are chosen as a random half of control units within each cluster. 

The rejection rates under the two scenarios described above are shown in Figure~\ref{fig:bf}. Our simulation results 
are based on $R = 1,000$ distinct randomizations and 100 independent replications of the data generating process.
Under scenario (A), we observe that all methods exhibit similar power and maintain validity at the 5\% level when the spillover effect is absent. Under scenario (B), however, \mlfrt\ is the most powerful method by a wide margin, as it is apparently capable of capturing the underlying nonlinearity in the spillover effect. \lmfrt\ loses power in this setting as it relies on linearity, and so does the \texttt{ELC}-based test that concerns a simple difference in means between exposed and non-exposed focal units.

\begin{figure}[t!]
    \vspace{-7mm}
    \centering
    \subfloat[Constant spillover]{\includegraphics[width=.52\linewidth]{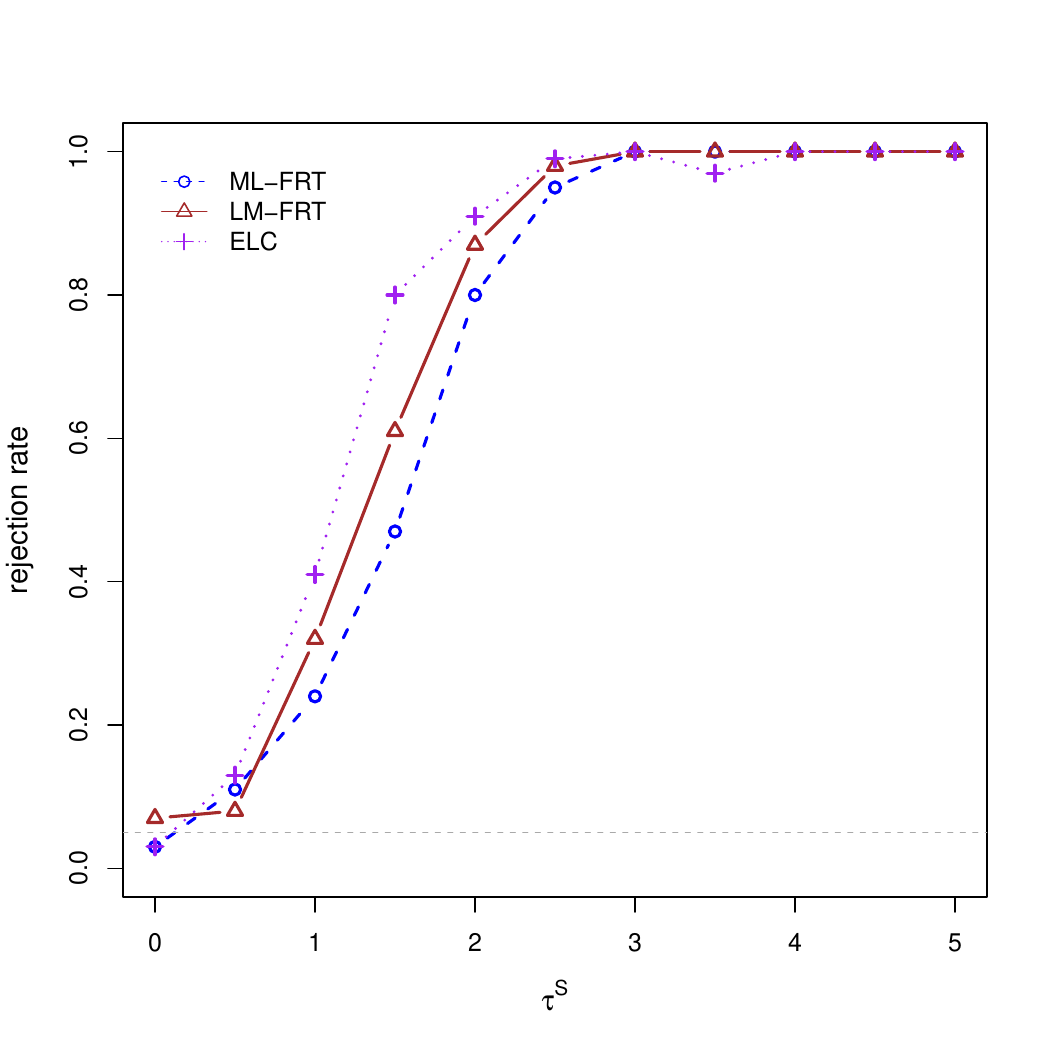}}
    \subfloat[Nonlinear spillover]{{\includegraphics[width=.51\linewidth]{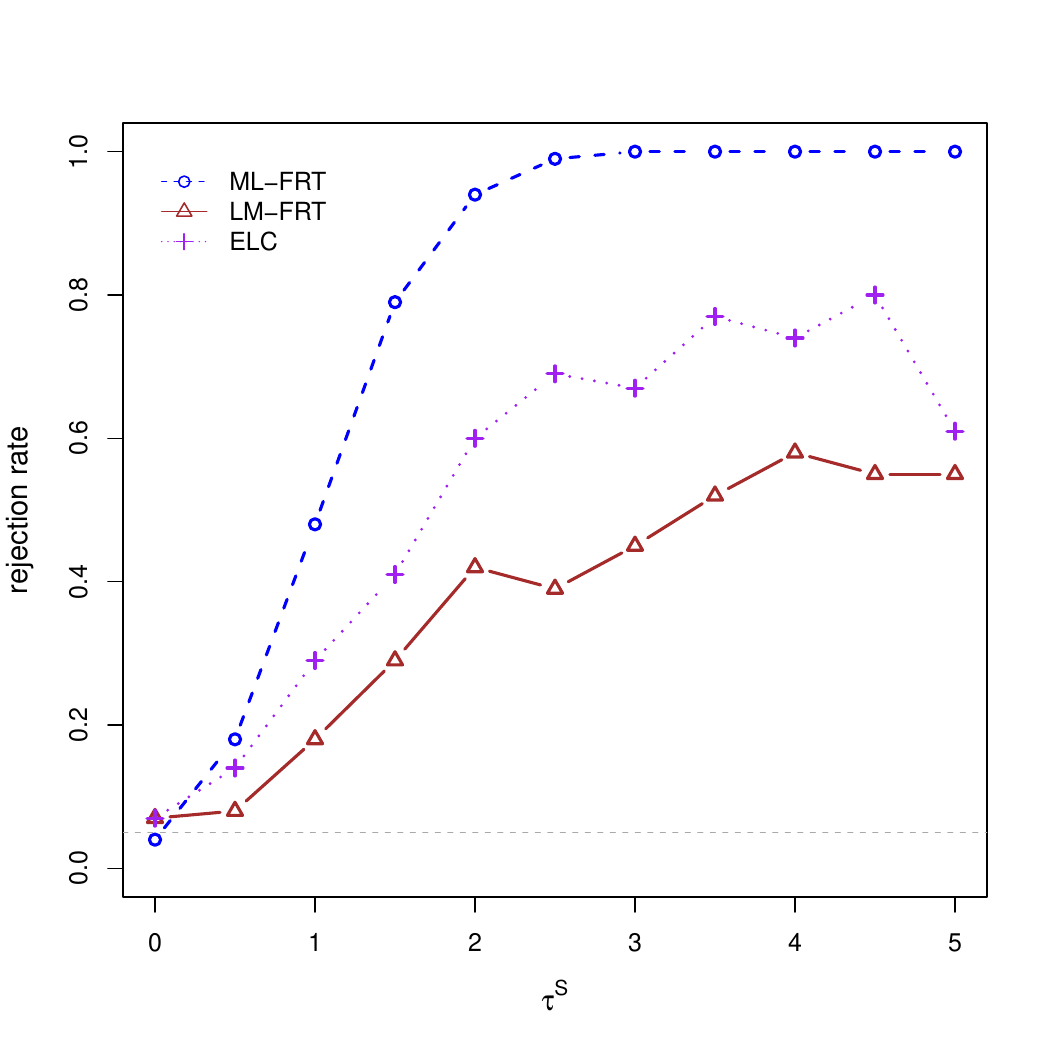}}}
    \caption{Rejection rates for spillover effects across $\tau^S$.}
    \label{fig:bf}
\end{figure}

\section{Conclusion}\label{sec:discussion}
In this paper, we introduced ML-assisted randomization tests to test for complex effects in causal inference. Our tests construct $p$-values by leveraging the true variation in the experiments, and thus are guaranteed to be finite-sample valid. This is one of the key advantage of our methods, as other existing tests based on flexible ML models are only valid in an asymptotic sense. Moreover, by leveraging the predictive power of ML models, our tests have superior power for capturing nonlinear, complex effects. This is demonstrated through our power theory and extensive numerical simulations.
However, our work leaves certain issues unexplored, opening up interesting problems for future work. 
For instance, it would be interesting to fuse randomization methods and causal ML further, by allowing randomization inference for causal ML estimators. 
Extending this synergy to more general designs ---e.g., factorial experiments--- would also have a high practical value.

\bibliographystyle{apalike}
\bibliography{references.bib}

\appendix

\newpage

\section{Simulation Details}
\subsection{Testing Constant Treatment Effects}\label{sec:simu_const}
Here we study the randomization test proposed in Section~\ref{sec:general}, which is designed to test whether the treatments have any effects on the outcomes. Based on the outcome model~\eqref{eq:our_model}, we consider three setups consisting of a constant treatment effect and a baseline effect that is potentially complex. 
\begin{itemize}
    \item $n = 200$, $p = 5$, and $X_i \stackrel{iid}{\sim} \calN(0, \Sigma)$, where $\Sigma$ is a randomly generated correlation matrix. We set $b(X_i) = \frac{1}{2} X_i^\top \beta$, $h(X_i) = \tau$ with $\tau \in \{0, 1, \dots, 10\}$, $g = 0$, and $\epsilon_i\sim\calN(0, 2^2)$, where $\beta\sim \mathrm{U}([1,30]^d)$. 
    \item $n = 200$, $p=1$, and $X_i\stackrel{iid}{\sim}\calN(0, 2^2)$. We set $b(X_i) = 2\mathbbm{1}\{X_{i1}<0.5\} - 3 \mathbbm{1}\{X_{i2}>-0.5\}$, $h(X_i) = \tau$ with $\tau \in \{0, 0.2, \dots, 1\}$,  $g = 0$, and $\epsilon_i \sim\calN(0, 0.1^2)$.
    \item $n = 200$, $p=1$, and $X_i\stackrel{iid}{\sim}\calN(0, 2^2)$. We set $b(X_i) = 2\cos(X_i)$, $h(X_i) = \tau$ with $\tau \in \{0, 0.2, \dots, 1\}$,  $g = 0$, and $\epsilon_i \sim\calN(0, 0.1^2)$.
\end{itemize}
Loosely speaking, we inspect different baseline effects induced by linear, piecewise-constant, and cosine functions. Since the second and third baseline effects are nonlinear, they are more difficult to detect using classical methods.

We implement our ML-assisted randomization test using both random forests ($\mlfrt$) and linear models ($\texttt{LM-FRT}$). For comparison purposes, we also evaluate standard FRTs where the test statistic is studentized Neyman's estimator ($\texttt{Neyman}$) and studentized Lin's estimator ($\texttt{Lin}$) \citep{zhao2021}, respectively. Neyman's estimator is the simple difference-in-means, whereas Lin's estimator is a linear-regression-adjusted ATE estimator.

Figure~\ref{fig:const_eff}(a)-(c) visualizes the rejection rates of different methods under the three setups above, which is based on $R = 100$ and 1,000 repeats.
First, observe that all methods have correct type I error control, since they all follow the framework of FRT and are thus finite-sample valid. However, different methods reveal different test power. In Figure~\ref{fig:const_eff}(a), $\texttt{LM-FRT}$ and $\texttt{Lin}$ achieve highest power, because the linear-model-based test statistics can better adjust for the linear baseline effects. In Figure~\ref{fig:const_eff}(b) and (c), $\mlfrt$ showcases highest power, as the flexible machine learning can better capture the nonlinear baseline effects.


\begin{figure}[htbp!]
    \vspace{-7mm}
    \centering
    \subfloat[Linear]{{\includegraphics[width=.33\linewidth]{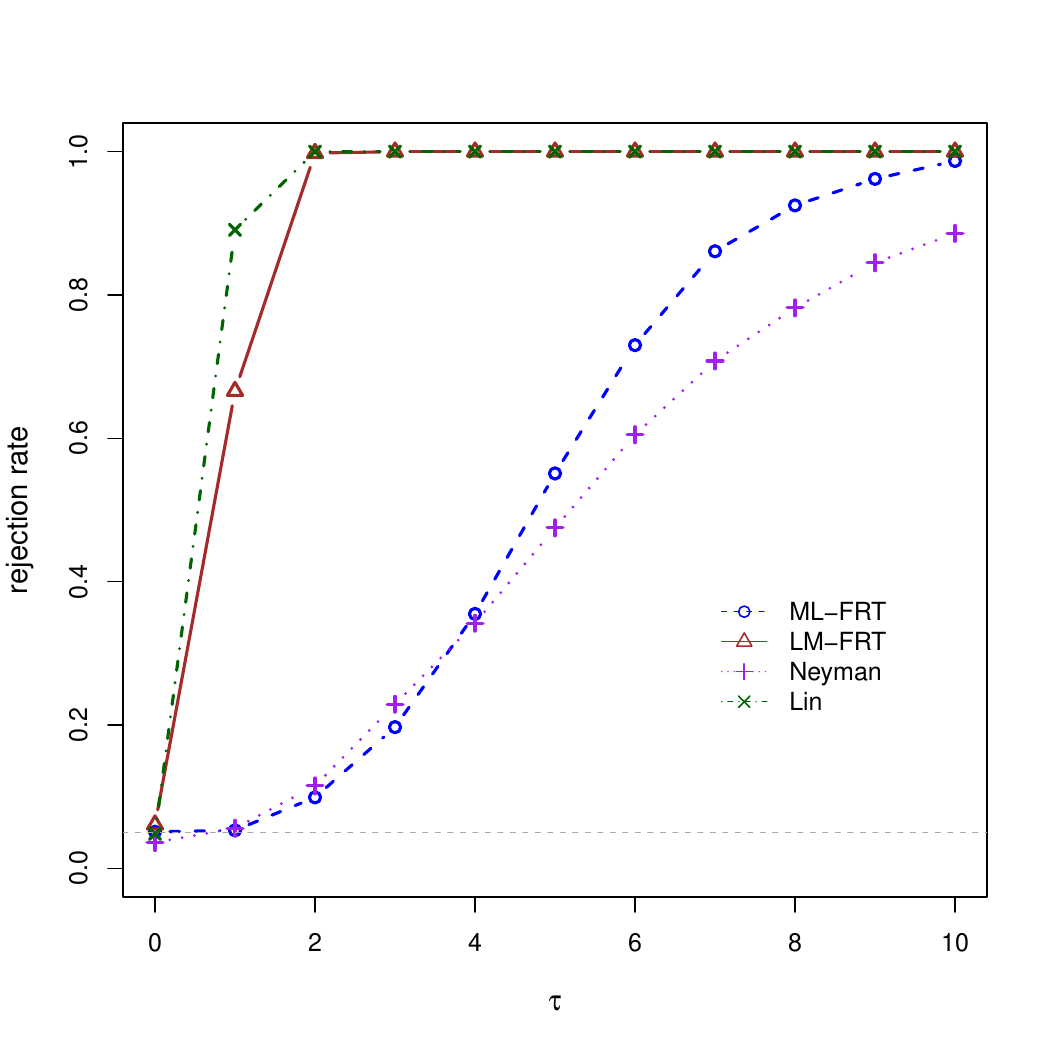}}}%
    \subfloat[Piecewise-constant]{{\includegraphics[width=.33\linewidth]{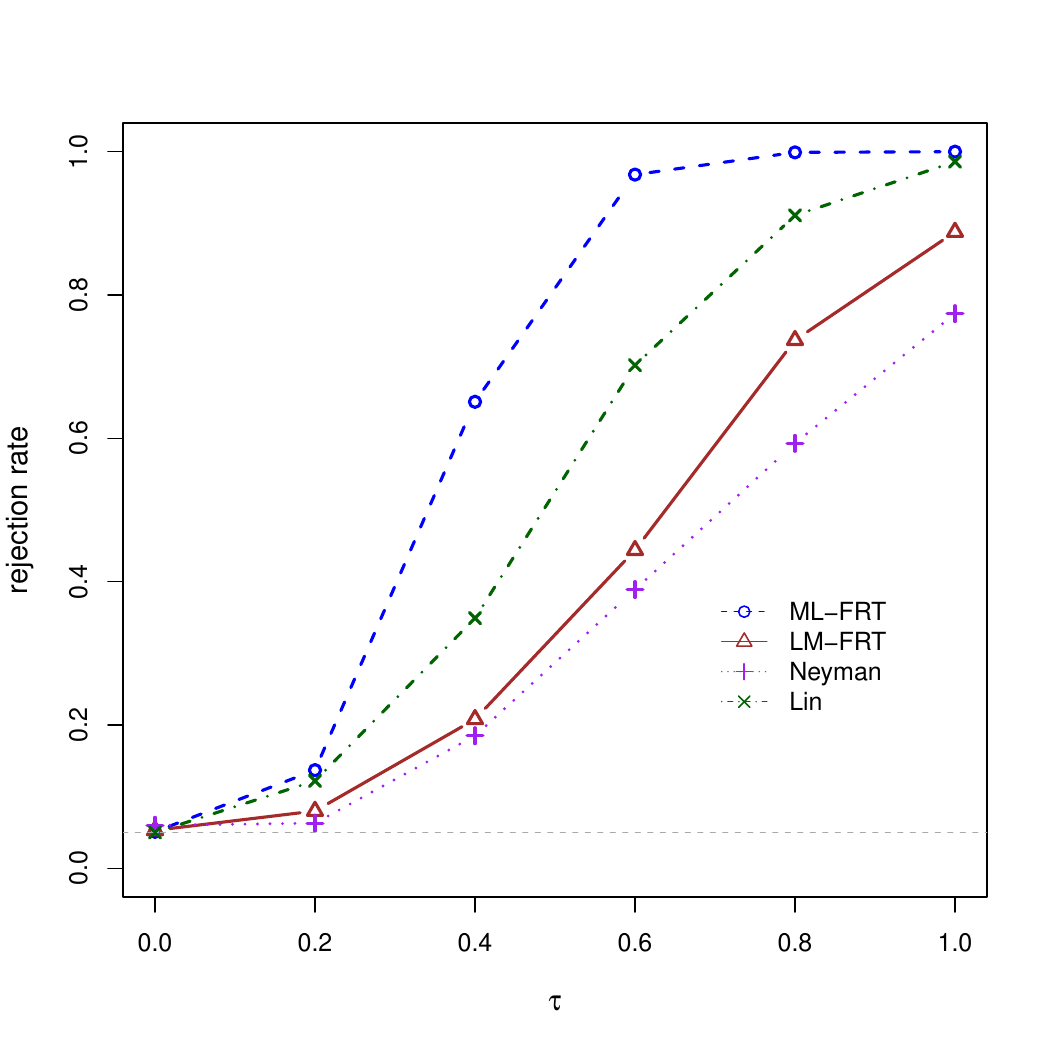}}}%
    \subfloat[Cosine]{{\includegraphics[width=.33\linewidth]{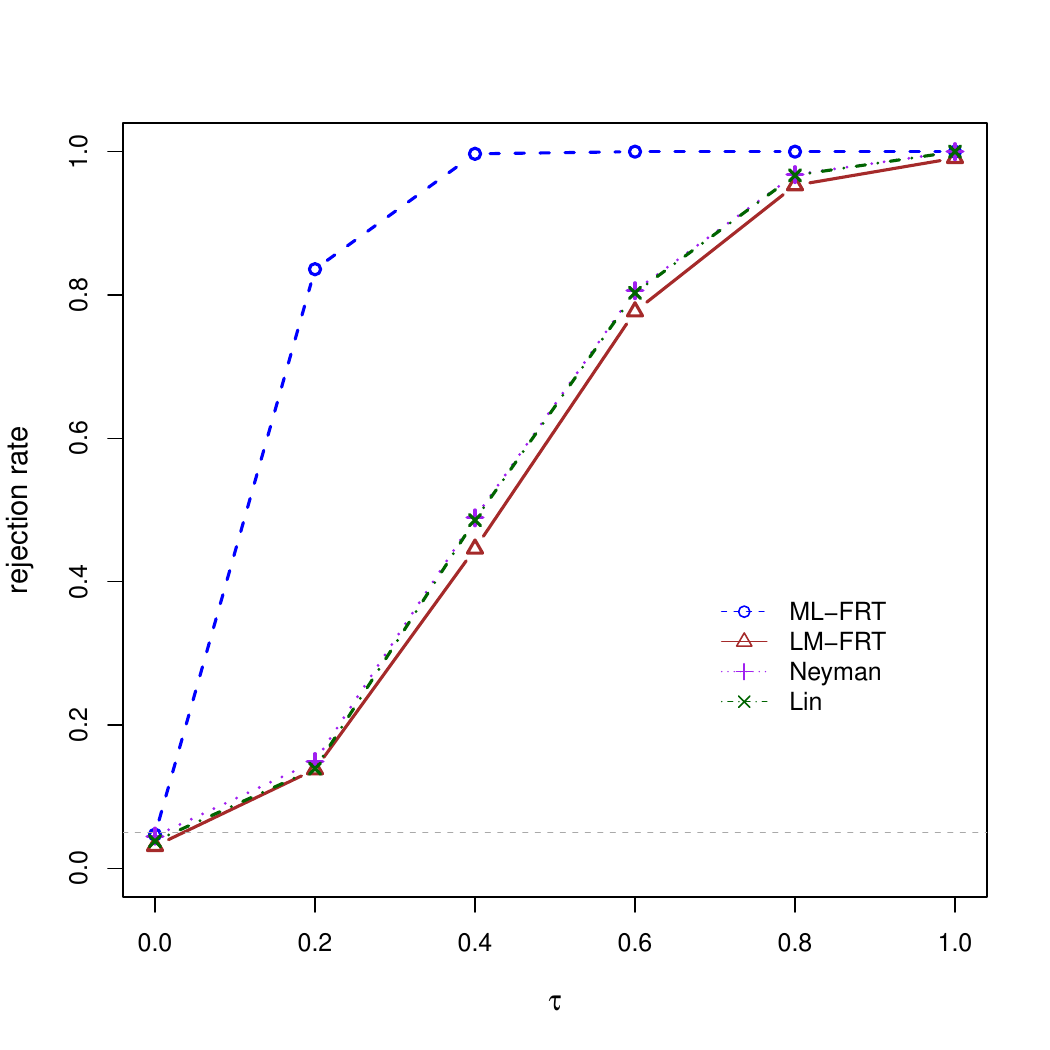}}}%
    \caption{Rejection rates for testing constant treatment effects with different baseline effects $b$.}
    \label{fig:const_eff}
\end{figure}

\subsection{Programming Details}
All numerical experiments in our paper are programmed in R. In $\mlfrt$, we employ the R library $\texttt{caret}$ \citep{caret2008} to compute cross-validation errors. In each cross-validation step, we use two R libraries $\texttt{randomForest}$ \citep{rf2002} and $\texttt{ranger}$ \citep{ranger2017} to fit random forests on the training data in different experiments. The number of folds ranges from 2 to 10 in different experiments, and empirically this choice does not have a significant effect on the performance of our tests. In addition, we inspect different tuning parameters in CV and choose the minimum CV error over all possible tuning parameters as our final CV error in the test statistic. The inspected tuning parameters include $\texttt{mtry}$, $\texttt{splitrule}$, and $\texttt{min.node.size}$, which controls number of variables to possibly split at in each node, splitting rule, and minimal node size to split at, respectively. 

Computing randomization $p$-values is computationally challenging, but can be easily parallelized. In our implementation, we run all simulations using parallel computation on computing clusters, with $\texttt{mem-per-cpu=}$8G and $\texttt{cpus-per-task=}$16, i.e., 16 parallel computing cpus with 8G memory. Under this setup, each simulation (i.e., each figure) in our simulation section takes around 15h to complete.

\section{Proofs on the Test Power}\label{sec:proof}
Recall the outcome model 
\begin{equation*}
    Y_i = \mu + b(X_i) + Z_i h(X_i) + g(\XX, Z_{-i}) + \epsilon_i.
\end{equation*}
In this section, we prove power results for testing
\begin{align*}
H_0^{\glob}: h=0, g=0 \quad \text{v.s.} \quad H_1^{\glob}: h \neq 0, g=0.
\end{align*}


For any $f_0\in\cF_0$ and $f\in\cF$, we define the squared loss function 
\begin{equation*}
    l(y, x; f_0) = (y - f_0(x))^2,\quad l(y, x,z; f) = (y - f(x,z))^2.
\end{equation*}
When there is no ambiguity, we write $l(f_0) = l(y, x; f)$ and $l(f) = l(y, x, z; f)$ for simplicity, that is, functions with subscript 0 are defined on $\cX$, whereas functions without subscript 0 are defined on $\cX\times \{0, 1\}$. To simplify our notation, we define $\cF = \cF_1$. 

Denote by $\E$ and $\E_n$ the population mean and empirical mean, respectively. Then we define the population loss and empirical loss as 
\begin{align*}
    \E l(f_0) &= \E (Y - f_0(X))^2, \quad \E l(f) = \E (Y - f(X, Z))^2.\\
    \E_n l(f_0) &= \frac{1}{n} \sum_{i=1}^n (Y_i - f_0(X_i))^2, \quad \E_n l(f) = \frac{1}{n} \sum_{i=1}^n (Y_i - f(X_i,Z_i))^2.
\end{align*}
In the definition above, $(Y, X, Z)$ is an independent copy of the observed data $(Y_i, X_i, Z_i)$ under Assumption~\ref{asmp:iid}. Let $[n] = \{1, \dots, n\}$.

\paragraph{Additional regularity conditions. }
We assume that the fitted models minimize the empirical loss in each cross-validation step: for a training data $(Y_i, X_i, Z_i)$ with $i\in D$, we have
\begin{align*}
    \widehat{f}_0 \in \underset{f_0\in\cF_0}{\arg\min} \sum_{i\in D} (\YY_i - f_0(X_i))^2,
    \quad \widehat{f} \in \underset{f\in\cF}{\arg\min} \sum_{i\in D} (\YY_i - f(X_i, \ZZ_i))^2.
\end{align*}
In words, the fitted models are empirical risk minimizers, which is commonly assumed for bounding the empirical risk of machine learning models \citep{anthony1999neural}. In practice, this condition holds for convex loss functions. On the flip side, however, this condition can be hard to verify for ML models that involve non-convex optimizations, such as neural networks.

In addition, we assume that $R > 1/\alpha - 1$, where $R$ is the number of randomizations in our main procedure. This is a mild assumption as $R$ is usually set to be a large constant, e.g., $R = 1,000$, such that $R > 1/\alpha - 1$ holds.


\subsection{Preliminaries}
In this section, we introduce some notations and preliminary results, which will be used in the main proof.

In cross-validation, we denote the $k$ folds by $D_1, \dots, D_k$ such that $D_i\cap D_j=\emptyset, |D_i| = |D_j|$ for $i\neq j$ and $\cup_j D_j = [n]$. Without loss of generality, suppose that the folds $D_1, \dots, D_k$ are deterministic. All of our theoretical results work for random folds as well, but one needs to apply the proof by first conditioning on the randomly selected folds. 

Given $(D_j)_{j=1}^k$, the cross-validation errors in our test statistic can be written as 
\begin{align*}
    \CV(\YY;\XX) &= \frac{1}{k} \sum_{j=1}^k \frac{1}{|D_j|} \sum_{i\in D_j} \left(Y_i - \widehat{f}_{0, -j}(X_i)\right)^2 = \frac{1}{k} \sum_{j=1}^k \E_{D_j} l(\widehat{f}_{0,-j}),\\
    \CV(\YY;\XX, \ZZ) &= \frac{1}{k} \sum_{j=1}^k \frac{1}{|D_j|} \sum_{i\in D_j} \left(Y_i - \widehat{f}_{-j}(X_i, Z_i)\right)^2 = \frac{1}{k} \sum_{j=1}^k \E_{D_j} l(\widehat{f}_{-j}),\\
    \CV(\YY;\XX, \ZZ^{(r)}) &= \frac{1}{k} \sum_{j=1}^k \frac{1}{|D_j|} \sum_{i\in D_j} \left(Y_i - \widehat{f}_{-j}(X_i, Z_i^{(r)})\right)^2 = \frac{1}{k} \sum_{j=1}^k \E_{D_j} l(\widehat{f}_{-j}^{(r)}).
\end{align*}
Here, $\widehat{f}_{0, -j}$, $\widehat{f}_{-j}$, $\widehat{f}_{-j}^{(r)}$ are fitted models on $j$-th training set $D_j^\complement\coloneqq [n]\backslash D_j$, and $E_{D_j}$ denotes the empirical mean over the samples in $D_j$. 

Assumptions~\ref{asmp:iid} and \ref{asmp:func_class} imply there exists a constant $M_0>0$ such that $l(f_0) \le M_0$, since
\begin{equation*}
    l(f_0) = (Y - f_0(X))^2 \le 2Y^2 + 2f_0^2(X) \le 4 M^2,
\end{equation*}
such that we may define $M_0 = 4M^2$. Same bound can be established for $l(f)$. Therefore, the loss functions are bounded by $M_0$ in our setup, a fact that we will use multiple times in our proof. 

We introduce concentration inequalities on $\E_{D_j} l(f_0)$ and $\E_{D_j} l(f)$, which comes directly from Hoeffding's inequality and the boundedness of the loss functions. 
\begin{lemma}\label{lem:hoeffding}
Suppose Assumptions~\ref{asmp:iid} and \ref{asmp:func_class} hold with the boundedness constant $M$. For any $j\in [k]$ and any $f_0\in\cF_0, f\in\cF$, we have
\begin{align*}
    \P\left( |\E_{D_j} l(Y, X; f_0) - \E l(Y, X;f_0)| \ge t\right) &\le 2\exp\left(-\frac{2 n t^2}{kM_0^2}\right),\\
    \P\left( |\E_{D_j} l(Y,X,Z;f) - \E l(Y,X,Z; f)| \ge t\right) &\le 2\exp\left(-\frac{2 n t^2}{kM_0^2}\right),\\
    \P\left( |\E_{D_j} l(Y,X,Z^{(r)};f) - \E l(Y,X,Z^{(r)}; f)| \ge t\right) &\le 2\exp\left(-\frac{2 nt^2}{k M_0^2}\right),
\end{align*}
for $t>0$. Here the randomness comes from $(Y_i, X_i, Z_i, Z_i^{(r)})$ for $i\in D_j$. 
\end{lemma}
\begin{proof}
Notice that under Assumptions~\ref{asmp:iid} and \ref{asmp:func_class}, we have
\begin{equation*}
    l(f_0) = (Y - f_0(X))^2 \le M_0
\end{equation*}
with probability one. Then, the quantity $\E_{D_j} l(f_0) - \E l(f_0)$ is an average of bounded random variables with mean zero. Then, the first bound follows from Hoeffding's inequality and $|D_j| = n/k$. The second and third bounds can be proved in a similar way. 
\end{proof}

Based on Lemma~\ref{lem:hoeffding}, we obtain the following result. 
\begin{lemma}\label{lem:cvloss}
Suppose Assumptions~\ref{asmp:iid} and \ref{asmp:func_class} hold. For any $j\in[k]$ and $t>0$, we have
\begin{align*}
    \P\left( |\E_{D_j} l(\widehat{f}_{0,-j}) - \E l(\widehat{f}_{0,-j})| \ge t\right) &\le 2\exp\left(-\frac{2 n t^2}{k M_0^2}\right),\\
    \P\left( |\E_{D_j} l(\widehat{f}_{-j}) - \E l(\widehat{f}_{-j})| \ge t\right) &\le 2\exp\left(-\frac{2 n t^2}{k M_0^2}\right),\\
    \P\left( |\E_{D_j} l(\widehat{f}_{-j}^{(r)}) - \E l(\widehat{f}_{-j}^{(r)})| \ge t\right) &\le 2\exp\left(-\frac{2 n t^2}{k M_0^2}\right).
\end{align*}
\end{lemma}
\begin{proof}
The proofs of three inequalities are identical and thus we only give the proof for $\widehat{f}_{-j}$. Since $(Y_i, X_i, Z_i)_{i\in D_j}$ and $\widehat{f}_{-j}$ are independent under Assumption~\ref{asmp:iid}, we can apply Lemma~\ref{lem:hoeffding} to obtain that with probability one, 
\begin{align*}
    \P\left( |\E_{D_j} l(\widehat{f}_{-j}) - \E l(\widehat{f}_{-j})| \ge t \mid \widehat{f}_{-j} \right) &\le 2\exp\left(-\frac{2 n t^2}{kM_0^2}\right),
\end{align*}
It implies the marginal concentration inequality, 
\begin{equation*}
    \P\left( |\E_{D_j} l(\widehat{f}_{-j}) - \E l(\widehat{f}_{-j})| \ge t  \right) \le 2\exp\left(-\frac{2 n t^2}{k M_0^2}\right).
\end{equation*}

\end{proof}

As the last piece of preliminaries, we introduce concentration bounds that measures the prediction error, which is adapted from the statistical learning theory. Define
\begin{align*}
    f_{0}^* &\in \underset{f_0\in\cF_0}{\arg\min}~\E l(Y,X; f_0), \\
    f^* &\in \underset{f\in\cF}{\arg\min}~\E l(Y,X,Z; f), \\
    f_{r}^* &\in \underset{f\in\cF}{\arg\min}~\E l(Y,X,Z^{(r)}; f).
\end{align*}
In words, $f_0^*$, $f^*$, and $f_r^*$ are best predictors with respect to certain data distributions. In the definition above, $Z^{(r)}$ is an independent copy of $Z$, which will be used to analyze the randomized statistic. 
\begin{lemma}\label{lem:rademacher}
Suppose Assumptions~\ref{asmp:iid} and \ref{asmp:func_class} hold. Then for any $t>0$, we have 
\begin{align*}
    \P\left( \E l(\widehat{f}_{0, -j}) - \E l(f_0^*) > 4\cR_{n-n/k}(\cF_{0};\P) + 2t \right) \le \exp\left(-\frac{(k-1)nt^2}{2kM_0^2}\right),\\
    \P\left( \E l(\widehat{f}_{-j}) - \E l(f^*) > 4\cR_{n-n/k}(\cF;\P) + 2t \right) \le \exp\left(-\frac{(k-1)nt^2}{2kM_0^2}\right),\\
    \P\left( \E l(\widehat{f}_{-j}^{(r)}) - \E l(f_r^*) > 4\cR_{n-n/k}(\cF;\P^{(r)}) + 2t \right) \le \exp\left(-\frac{(k-1)nt^2}{2kM_0^2}\right).
\end{align*}
where $\cR_{n-n/k}(\cF_0;\P)$, $\cR_{n-n/k}(\cF;\P)$, and $\cR_{n-n/k}(\cF;\P^{(r)})$ are Rademacher complexities defined by \eqref{eq:rademacher_cplx}.
\end{lemma}
\begin{proof}
The proofs for different models are identical, so here we only give a proof for $\widehat{f}_{-j}$. To analyze the excess loss $\E l(\widehat{f}_{-j}) - \E l(f_*)$, notice that 
\begin{align*}
    \E l(\widehat{f}_{-j}) - \E l(f_*) &= \E l(\widehat{f}_{-j}) - \E_{D_j^\complement} l(\widehat{f}_{-j}) + \E_{D_j^\complement} l(\widehat{f}_{-j}) - \E_{D_j^\complement} l(f_*)+ \E_{D_j^\complement} l(f_*)- \E l(f_*) \\
    &\stackrel{\text{(i)}}{\le} \E l(\widehat{f}_{-j}) - \E_{D_j^\complement} l(\widehat{f}_{-j}) + \E_{D_j^\complement} l(f_*)- \E l(f_*) \\
    &\le 2 \sup_{f\in\cF} |\E_{D_j^\complement} l(f) - \E l(f)|,
\end{align*}
where (i) uses the fact that $\E_{D_j^\complement} l(\widehat{f}_{-j}) - \E_{D_j^\complement} l(f_*)\le 0$, since $\widehat{f}_{-j}$ minimizes the empirical loss $\E_{D_j^\complement} l(f)$. Since $l(f)$ is uniformly bounded by $M_0$ for any $(Y,X,Z)$ and $f\in\cF$, and hence we we apply Theorem 4.10 of \cite{wainwright2019high} to obtain
\begin{equation*}
    \sup_{f} |\E_{D_j^\complement} l(f) - \E l(f)|\le 2\cR_{|D_j^\complement|}(\cF;\P) + t
\end{equation*}
with probability at least $1 - \exp(-|D_j^\complement|t^2/2M_0^2)$. Hence,
\begin{equation*}
    \E l(\widehat{f}_{-j}) - \E l(f_*) \le 2\sup_{f} |\E_{D_j^\complement} l(f) - \E l(f)|\le 4\cR_{|D_j^\complement|}(\cF;\P) + 2t
\end{equation*}
with probability at least $1 - \exp(-|D_j^\complement|t^2/2M_0^2)$. That is, 
\begin{equation*}
\P\left( \E l(\widehat{f}_{-j}) - \E l(f^*) > 4\cR_{n-n/k}(\cF;\P) + 2t \right) \le \exp(-|D_j^\complement|t^2/2M_0^2) =  \exp\left(-\frac{(k-1)nt^2}{2kM_0^2}\right),
\end{equation*}
where the last equality follows from $|D_j^\complement| = n - n/k$.
\end{proof}
In words, with high probability, the prediction errors between fitted predictor and the best in-class predictor can be upper bounded by a certain Rademacher complexity. 

\subsection{Proof of Theorem~\ref{thm:power}}
First, we establish a key lemma on the type II error, which will be the main tool for proving Theorem~\ref{thm:power}.
\begin{lemma}\label{lem:power}
Suppose Assumption~\ref{asmp:iid} and Condition 1 in Assumption~\ref{asmp:func_class} hold with a boundedness constant $M>1$. Define 
\begin{multline*}
    L = \inf_{f\in \cF} \E (Y - f(X, Z^{(r)}))^2 - \inf_{f\in \cF} \E (Y - f(X, Z))^2 \\
    - 8\cR_{n-n/k}(\cF_{0};\P) - 4\cR_{n-n/k}(\cF;\P) - 4\cR_{n-n/k}(\cF;\P^{(r)}).
\end{multline*}
Then, under the alternative $H_1^{\glob}$, if $L > 0$, the type II error of our randomization test satisfies 
\begin{align*}
    \P(\mathrm{pval}>\alpha) &\le 4R \left(2k \exp\left(-\frac{n L^2}{32 k M_0^2}\right) + \exp\left(-\frac{(k-1)nL^2}{128 kM_0^2}\right)\right)\\
    &= O\left( k \exp\left(-\frac{C n L^2}{k M^4}\right) \right)
\end{align*}
for some small constant $C$.
\end{lemma}
\begin{proof}[Proof of Lemma~\ref{lem:power}]
First, notice that
\begin{equation}\label{eq:type2_bound1}
    \begin{aligned}
    \P(\mathrm{pval} > \alpha) &= \P\left(\frac{1}{R+1}\left(\sum_{r = 1}^{R} \mathbbm{1}\{t_n(\YY, \ZZ^{(r)}, \XX)>t_n(\YY, \ZZ, \XX)\} + U(1+m_R)\right) > \alpha\right)\\
    &\stackrel{\text{(i)}}{\le} \P\left(\sum_{r = 1}^{R} \mathbbm{1}\{t_n(\YY, \ZZ^{(r)}, \XX)>t_n(\YY, \ZZ, \XX)\} + 1> (R+1)\alpha\right)\\
    &\stackrel{\text{(ii)}}{\le} \P\left(\exists r, t_n(\YY, \ZZ^{(r)}, \XX)>t_n(\YY, \ZZ, \XX)\right),
\end{aligned}
\end{equation}
where (i) follows by the assumption that $m_R=0$ with probability one, and (ii) follows from $(R+1)\alpha>1$. Then we apply a union bound to obtain
\begin{equation}\label{eq:type2_bound2}
    \begin{aligned}
    \P\left(\exists r, t_n(\YY, \ZZ^{(r)}, \XX)>t_n(\YY, \ZZ, \XX)\right) &\le \sum_{r = 1}^{R} \P\left(t_n(\YY, \ZZ^{(r)}, \XX)>t_n(\YY, \ZZ, \XX)\right) \\
    &\stackrel{\text{(i)}}{=} R\cdot\P\left(t_n(\YY, \ZZ^{(r)}, \XX)>t_n(\YY, \ZZ, \XX)\right)\\
    &\stackrel{\text{(ii)}}{\le} R\Bigl(\P\bigl(t_n(\YY, \ZZ^{(r)}, \XX)\ge c\bigr) + \P\bigl(t_n(\YY, \ZZ, \XX)\le c\bigr)\Bigr).
\end{aligned}
\end{equation}
In the derivation above, (i) uses the fact that $\ZZ^{(r)}$ are i.i.d.; (ii) follows from the identity $\mathbbm{1}\{X > Y\} \le \mathbbm{1}\{X\ge c\} + \mathbbm{1}\{Y\le c\}$. Based on \eqref{eq:type2_bound1} and \eqref{eq:type2_bound2}, we have shown that 
\begin{equation*}
    \text{type II error} = \P(\mathrm{pval} > \alpha) \le R (\underbrace{F_{t_n}(c)}_\text{CDF of $t_n(Y,Z,\bX)$} + 1 - \underbrace{F_{t_n^{(r)}}(c)}_\text{CDF of $t_n(\YY, \ZZ^{(r)}, \XX)$}).
\end{equation*}
That is, it suffices to analyze the sampling distribution $t_n(\YY, \ZZ, \XX)$, and the randomization distribution $t_n(\YY, \ZZ^{(r)}, \XX)$.

\noindent \underline{Step 1. Analyze $t_n(\YY, \ZZ, \XX)$.}
Here, we show that the sampling distribution $t_n(\YY, \ZZ, \XX)$ concentrates around a constant $\Delta_0$. For $t_n(\YY, \ZZ, \XX) = \CV(\YY; \XX) - \CV(\YY;\XX, \ZZ)$, notice that 
\begin{align*}
    \CV(\YY; \XX) &= \frac{1}{k}\sum_{j=1}^k \E_{D_j} l(\widehat{f}_{0, -j}), \quad \CV(\YY;\XX, \ZZ) = \frac{1}{k}\sum_{j=1}^k \E_{D_j} l(\widehat{f}_{-j}),\\
    &\Rightarrow t_n(\YY, \ZZ, \XX) = \frac{1}{k}\sum_{j=1}^k \E_{D_j} \left(l(\widehat{f}_{0, -j}) - l(\widehat{f}_{-j})\right).
\end{align*}

As the first step, we reduce the empirical loss to the population loss through Lemma~\ref{lem:cvloss}. Specifically, notice that 
\begin{equation}\label{eq:tobs_step1}
    \begin{aligned}
    t_n(\YY, \ZZ, \XX) &= \frac{1}{k}\sum_{j=1}^k \E \left(l(\widehat{f}_{0, -j}) - l(\widehat{f}_{-j})\right) \\
    &+ \frac{1}{k}\sum_{j=1}^k \left( \E_{D_j} l(\widehat{f}_{0, -j}) - \E l(\widehat{f}_{0,-j})\right) + \frac{1}{k}\sum_{j=1}^k \left( \E l(\widehat{f}_{-j}) - \E_{D_j} l(\widehat{f}_{-j})\right).
\end{aligned}
\end{equation}
By applying Lemma~\ref{lem:cvloss} to $\E_{D_j}l(\widehat{f}_{0, -j})$ and $\E_{D_j} l(\widehat{f}_{-j})$ separately, we obtain
\begin{equation}\label{eq:cvloss}
    |\E_{D_j} l(\widehat{f}_{0, -j}) - \E l (\widehat{f}_{0, -j})|  < t, \quad |\E_{D_j} l(\widehat{f}_{-j}) - \E l (\widehat{f}_{-j})|  < t, 
\end{equation}
each of which holds with probability at least $1 - 2\exp\left(-\frac{2 n t^2}{k M_0^2}\right)$. Then, we apply \eqref{eq:cvloss} to \eqref{eq:tobs_step1} and obtain
\begin{align*}
    &\P\left(\left|t_n(\YY, \ZZ, \XX) -  \frac{1}{k}\sum_{j=1}^k \E \left(l(\widehat{f}_{0, -j}) - l(\widehat{f}_{-j})\right)\right| > 2t \right) \\
    &= \P\left(\left|\sum_{j=1}^k \left( \E_{D_j} l(\widehat{f}_{0, -j}) - \E l(\widehat{f}_{0,-j})\right) + \sum_{j=1}^k \left( \E l(\widehat{f}_{-j}) - \E_{D_j} l(\widehat{f}_{-j})\right)\right| > 2kt \right) \\
    &\stackrel{\text{(i)}}{\le} \sum_j \P\left(\left| \E_{D_j} l(\widehat{f}_{0, -j}) - \E l(\widehat{f}_{0,-j})\right| > t \right) + \sum_j \P\left(\left|\E l(\widehat{f}_{-j}) - \E_{D_j} l(\widehat{f}_{-j})\right| > t \right) \\
    &\le 4k \exp\left(-\frac{2 n t^2}{k M_0^2}\right)
\end{align*}
where (i) follows from a union bound. 

Since the data $(\YY_i, \ZZ_i, X_i)$ are exchangeable, we have $\E l(\widehat{f}_{0, -j}) = \E l(\widehat{f}_{0, -1})$ and $\E l(\widehat{f}_{-j}) = \E l(\widehat{f}_{-1})$. Therefore,
\begin{equation}\label{eq:tobs_result1}
    \left|t_n(\YY, \ZZ, \XX) -  \E \left(l(\widehat{f}_{0, -1}) - l(\widehat{f}_{-1})\right)\right| < 2t 
\end{equation}
with probability at least $1 - 4k \exp\left(-\frac{2n t^2}{k M_0^2}\right)$. 

As the second step, we connect the empirical optimizers with the population ones through Lemma~\ref{lem:rademacher}. Note that 
\begin{equation}\label{eq:tobs_step2}
    \begin{aligned}
    \E l(\widehat{f}_{0, -1}) - \E l(\widehat{f}_{-1}) &= \E l(f_{0}^*) - \E l(f^*)\\
    &+ \left(\E l(\widehat{f}_{0, -1}) - \E l({f}_{0}^*)\right) + \left(\E l(f^*) - \E l(\widehat{f}_{-1})\right).
\end{aligned}
\end{equation}
By applying Lemma~\ref{lem:rademacher} to $\widehat{f}_{0,-1}$ and $\widehat{f}_{-1}$ separately, we obtain
\begin{equation}\label{eq:rademacher}
\begin{aligned}
        0 \le \E l(\widehat{f}_{0, -1}) - \E l({f}_{0}^*) &< 4\cR_{n-n/k}(\cF_{0};\P) + 2t, \\
        0 \le \E l(\widehat{f}_{-1}) - \E l({f}^*) &< 4\cR_{n-n/k}(\cF;\P) + 2t,
\end{aligned}
\end{equation}
each of which holds with probability at least $1 - \exp\left(-\frac{(k-1)nt^2}{2kM_0^2}\right)$. Then, we apply \eqref{eq:rademacher} to \eqref{eq:tobs_step2} using a union bound, and obtain
\begin{equation}\label{eq:tobs_result2}
    \left|\E l(\widehat{f}_{0, -1}) - \E l(\widehat{f}_{-1}) - \left(\E l(f_{0}^*) - \E l(f^*)\right)\right| < 4\max\{\cR_{n-n/k}(\cF_{0};\P), \cR_{n-n/k}(\cF);\P\} + 2t
\end{equation}
with probability at least $1 - 2\exp\left(-\frac{(k-1)nt^2}{2kM_0^2}\right)$.

Combining two concentration bounds in \eqref{eq:tobs_result1} and \eqref{eq:tobs_result2}, we have 
\begin{align*}
    \left|t_n(\YY, \ZZ, \XX) - \Delta_0\right| &\le 4\max\{\cR_{n-n/k}(\cF_{0};\P), \cR_{n-n/k}(\cF;\P)\} + 4t, \\
    \Delta_0 &\coloneqq \E l(f_{0}^*) - \E l(f^*),
\end{align*}
with probability at least $1 - 4k \exp\left(-\frac{2 n t^2}{k M_0^2}\right) - 2\exp\left(-\frac{(k-1)nt^2}{2kM_0^2}\right)$. In the expression above, $\Delta_0$ is introduced to quantify the magnitude of $t_n(\YY, \ZZ, \XX)$. 

\noindent\underline{Step 2. Analyze $t_n(\YY, \ZZ^{(r)}, \XX)$.}
We show that the randomization distribution $t_n(\YY, \ZZ^{(r)}, \XX)$ concentrates around $\Delta_r$, which is potentially smaller than $\Delta_0$. Following a similar analysis as above, we first apply Lemma~\ref{lem:cvloss} to obtain that 
\begin{equation}\label{eq:trand_result1}
    \left|t_n(\YY, \ZZ^{(r)}, \XX) -  \E \left(l(\hat{f}_{0, -1}) - l(\hat{f}_{-1}^{(r)})\right)\right| < 2t 
\end{equation}
with probability at least $1 - 4k \exp\left(-\frac{2 n t^2}{k M_0^2}\right)$. Note that $\widehat{f}_{0, -j}$ does not depend on the randomized treatment assignments $Z^{(r)}$. 

Next we apply Lemma~\ref{lem:rademacher} to obtain that 
\begin{equation}\label{eq:trand_result2}
    \left|\E \left(l(\hat{f}_{0, -1}) - l(\hat{f}_{-1}^{(r)})\right) - \left(\E l(f_{0}^{*}) - \E l(f_r^{*})\right)\right| < 4\max\{\cR_{n-n/k}(\cF_{0};\P), \cR_{n-n/k}(\cF;\P^{(r)})\} + 2t
\end{equation}
with probability at least $1 - 2\exp\left(-\frac{(k-1)nt^2}{2kM_0^2}\right)$.

Combining two concentration bounds in \eqref{eq:trand_result1} and \eqref{eq:trand_result2}, we have 
\begin{align*}
    \left|t_n(\YY, \ZZ^{(r)}, \XX) - \Delta_r\right| &< 4\max\{\cR_{n-n/k}(\cF_{0};\P), \cR_{n-n/k}(\cF;\P^{(r)})\} + 4t,\\
    \Delta_r &\coloneqq \E l(f_{0}^*) - \E l(f_r^*)
\end{align*}
with probability at least $1 - 4k \exp\left(-\frac{2 n t^2}{k M_0^2}\right) - 2\exp\left(-\frac{(k-1)nt^2}{2kM_0^2}\right)$. In the expression above, $\Delta_r$ is introduced to quantify the magnitude of $t_n(\YY, \ZZ^{(r)}, \XX)$. Next we will quantify the difference between $\Delta_0$ and $\Delta_r$.

\noindent\underline{Step 3. Derive final results.}
Based on Step 1 above, $t_n(\YY,\ZZ,\XX)$ satisfies a high probability bound
\begin{align*}
    \P\left(t_n(\YY, \ZZ, \XX) < M_1- 4t\right) &\le 4k \exp\left(-\frac{2 n t^2}{k M_0^2}\right) + 2\exp\left(-\frac{(k-1)nt^2}{2kM_0^2}\right), \\
    M_1 &\coloneqq \Delta_0 - 4\max\{\cR_{n-n/k}(\cF_{0};\P), \cR_{n-n/k}(\cF;\P)\}
\end{align*}
for any $t>0$. Based on Step 2, $t_n(\YY, \ZZ^{(r)}, \XX)$ satisfies
\begin{align*}
    \P\left(t_n(\YY, \ZZ^{(r)}, \XX) > M_2 + 4t\right) &\le 4k \exp\left(-\frac{2 n t^2}{k M_0^2}\right) + 2\exp\left(-\frac{(k-1)nt^2}{2kM_0^2}\right), \\
    M_2 &\coloneqq \Delta_r + 4\max\{\cR_{n-n/k}(\cF_{0};\P), \cR_{n-n/k}(\cF;\P^{(r)})\}
\end{align*}
for any $t>0$. Recall that the type II error can be upper bounded by certain CDFs at a specified value $c$ \eqref{eq:type2_bound2}. Here, we set $c = (M_1 + M_2)/2$ and $L_0 = M_1 - M_2$. Note that 
\begin{align*}
 L_0 &= \Delta_0 - \Delta_r - 4\max\{\cR_{n-n/k}(\cF_{0};\P), \cR_{n-n/k}(\cF;\P)\} - 4\max\{\cR_{n-n/k}(\cF_{0};\P), \cR_{n-n/k}(\cF;\P^{(r)})\} \\
 &\ge \Delta_0 - \Delta_r - 8\cR_{n-n/k}(\cF_{0};\P) - 4\cR_{n-n/k}(\cF;\P) - 4\cR_{n-n/k}(\cF;\P^{(r)}) \eqqcolon L.
\end{align*}
Since $L>0$, we have $L_0>0$. We obtain 
\begin{align*}
    \P\bigl(t_n(\YY, \ZZ, \XX)\le c\bigr) & = \P\left(t_n(\YY, \ZZ, \XX)\le \frac{M_1 + M_2}{2}\right) \\
    &= \P\left(t_n(\YY, \ZZ, \XX)\le M_1 - 4 \frac{L_0}{8}\right) \\
    &\le 4k \exp\left(-\frac{2n L_0^2}{64 k M_0^2}\right) + 2\exp\left(-\frac{(k-1)nL_0^2}{2\times 64 kM_0^2}\right)\\
    &= 4k \exp\left(-\frac{n L_0^2}{32 k M_0^2}\right) + 2\exp\left(-\frac{(k-1)nL_0^2}{128 kM_0^2}\right).
\end{align*}
Similarly, 
\begin{align*}
    \P\bigl(t_n(\YY, \ZZ^{(r)}, \XX)\ge c\bigr) & = \P\left(t_n(\YY, \ZZ^{(r)}, \XX)\ge \frac{M_1 + M_2}{2}\right) \\
    &= \P\left(t_n(\YY, \ZZ^{(r)}, \XX)\ge M_2 + 4 \frac{L_0}{8}\right) \\
    &\le 4k \exp\left(-\frac{n L_0^2}{32 k M_0^2}\right) + 2\exp\left(-\frac{(k-1)nL_0^2}{128 kM_0^2}\right).
\end{align*}
Based on Equations~\eqref{eq:type2_bound1} and \eqref{eq:type2_bound2}, we can upper bound the type II error by
\begin{align*}
    \P (\mathrm{pval}>\alpha) \le 4R \left(2k \exp\left(-\frac{n L_0^2}{32 k M_0^2}\right) + \exp\left(-\frac{(k-1)nL_0^2}{128 kM_0^2}\right)\right)\\
    \le 4R \left(2k \exp\left(-\frac{n L^2}{32 k M_0^2}\right) + \exp\left(-\frac{(k-1)nL^2}{128 kM_0^2}\right)\right).
\end{align*}
In the inequality, we replace $L_0$ by $L$ since $L_0\ge L$.

To further simplify the result, note that for $k\ge 2$, 
\begin{equation*}
    \frac{k-1}{k} \ge \frac{1}{k}.
\end{equation*}
Therefore, 
\begin{align*}
    \P (\mathrm{pval}>\alpha)
    &\le 4R \left(2k \exp\left(-\frac{n L^2}{32 k M_0^2}\right) + \exp\left(-\frac{(k-1)nL^2}{128 kM_0^2}\right)\right)\\
    &\le 4R \left(2k \exp\left(-\frac{n L^2}{32 k M_0^2}\right) + \exp\left(-\frac{nL^2}{128 kM_0^2}\right)\right)\\
    &\le 4R (2k + 1) \exp\left(-\frac{nL^2}{128 kM_0^2}\right)\\
    &= O\left( k \exp\left(-\frac{nL^2}{128 k \times 16 M^4}\right) \right) = O\left( k \exp\left(-\frac{0.0004 n L^2}{k M^4}\right) \right).
\end{align*}
The last line uses the fact $M_0 = 4M^2$. This completes the proof.
\end{proof}

Next we prove Theorem~\ref{thm:power} in the main text using Lemma~\ref{lem:power}.
\begin{proof}[Proof of Theorem~\ref{thm:power}]
Note that $\Delta = \Delta_0 - \Delta_r$ by definiton. To apply Lemma~\ref{lem:power}, we need to verify
\begin{align*}
    L = \Delta_0 - \Delta_r - 8\cR_{n-n/k}(\cF_{0};\P) - 4\cR_{n-n/k}(\cF;\P) - 4\cR_{n-n/k}(\cF;\P^{(r)}) > 0.
\end{align*}

First we show that $\Delta_0 - \Delta_r \ge 0$ under Condition 2 of Assumption~\ref{asmp:func_class}. By definition
\begin{equation*}
    \Delta_0 - \Delta_r = \E l(f_r^*) - \E l(f^*) = \inf_f \E (Y 
- f(X, Z^{(r)}))^2 - \inf_f \E (Y - f(X, Z))^2.
\end{equation*}
To simplify the notation, for any $f\in\cF$, we define $\delta_0(x) = f(x,0)$ and $\delta(x) = f(x, 1) - f(x, 0)$. Then, we have
\begin{align*}
\E (Y - f(X, Z))^2 &= \E (\mu + b(X) + Z h(X) - \delta_0(X) - Z \delta(X))^2\\
&= \E (\mu + b(X) - \delta_0(X) + \pi (h(X)-\delta(X)) + (Z - \pi) (h(X) - \delta(X))^2\\
&= \E (\mu + b(X) - \delta_0(X) + \pi (h(X)-\delta(X)))^2 + \pi(1 - \pi) \E (h(X) - \delta(X))^2,
\end{align*}
where the last equality uses the fact $\E (Z-\pi)|X = 0$ from Assumption~\ref{asmp:iid}. Similarly, we have
\begin{align*}
   \E (Y - f(X, Z^r))^2 &= \E (\mu + b(X) + Z h(X) - \delta_0(X) - Z^r \delta(X))^2\\
&= \E (\mu + b(X) - \delta_0(X) + \pi (h(X)-\delta(X)) + (Z - \pi) h(X) + (Z^r - \pi) \delta(X))^2\\
&= \E (\mu + b(X) - \delta_0(X) + \pi (h(X)-\delta(X)))^2 + \pi(1 - \pi) \E (h^2(X) + \delta^2(X)), 
\end{align*}
where the last equality uses the fact $\E (Z-\pi)|X = \E (Z^r-\pi)|X= 0$ and $Z\indep Z^r|X$. Therefore, for any $f\in\cF$, 
\begin{equation*}
\E (Y - f(X, Z^r))^2 - \E (Y - f(X, Z))^2 = 2\pi(1-\pi) \E h(X)\delta(X).
\end{equation*}
Under Condition 2 of Assumption~\ref{asmp:func_class}, we have
\begin{equation*}
    \E h(X) (f_r^*(X,1) - f_r^*(X,0)) \ge 0.
\end{equation*}
Therefore, 
\begin{align*}
    \inf_f \E(Y - f(X,Z^{(r)}))^2 &= \E(Y - f_r^*(X,Z^{(r)}))^2\\
    &= \E(Y - f_r^*(X,Z))^2 + 2\pi(1-\pi)  \E h(X) (f_r^*(X,1) - f_r^*(X,0)) \\
    &\ge \E(Y - f_r^*(X,Z))^2 \ge \inf_f \E(Y - f(X,Z))^2.
\end{align*}
The second equality uses the fact that $\delta(x) = f(x, 1) - f(x, 0)$ for any $f$. Then, we obtain $\Delta = \Delta_0 - \Delta_r \ge 0$.

Then, under Condition 3 of Assumption~\ref{asmp:func_class}, the Rademacher complexities $\cR_n(\cF_0;\P)$, $\cR_n(\cF;\P)$, $\cR_n(\cF;\P^{(r)})$ are of order $o(1)$. Since $k$ is fixed, the condition above implies that
\begin{equation*}
    \cR_{n-n/k}(\cF_0;\P) = o(1), \quad \cR_{n - n/k}(\cF;\P) = o(1),\quad \cR_{n - n/k}(\cF;\P^{(r)}) = o(1).
\end{equation*}
That is, we have $L - \Delta = o(1)$. Therefore, under Assumptions~\ref{asmp:iid} and \ref{asmp:func_class}, if $\Delta > 0$, we can apply Lemma~\ref{lem:power} and replace $L$ by $\Delta$. This gives the type II error bound in Theorem~\ref{thm:power}. If $\Delta = 0$, the type II error bound in Theorem~\ref{thm:power} is greater than one, and thus automatically holds.
\end{proof}

\subsection{Proofs of Propositions~\ref{prop:res} and \ref{prop:relative_efficiency}}
In this section, we prove the power comparison between ML-FRT and the residualized approach (RES) in Section~\ref{sec:res}. Recall that the function classes $\cF_1$ and $\cF_0$ for RES are defined in Section~\ref{sec:res}, and we use $\cF = \cF_1$ for simplicity.
First we prove Proposition~\ref{prop:res}. 
\begin{proof}[Proof of Proposition~\ref{prop:res}]
In the proof, we focus on proving Theorem~\ref{thm:power} and deriving $\Delta$ for the residualized approach (RES). At a high level, we proceed in three steps.
\begin{enumerate}
    \item Apply Theorem~\ref{thm:power} to an idealized RES approach with fitted model $m^*$. 
    \item Establish asymptotic equivalence between the tests induced by $m^*$ and $\widehat{m}$, such that the result of the idealized RES in Step 1 can be translated to an error bound for the actual RES with $\widehat{m}$.
    \item Simplify $\Delta^{\res}$. 
\end{enumerate}

The result for ML-FRT can be obtained in a similar manner, which we discuss in the end of the proof.

\noindent \underline{Step 1. Prove Theorem~\ref{thm:power} for an ideal RES approach with $m^*$.}
We consider an ideal RES approach, which uses the best predictor $m^*$ to compute the residuals
\begin{align*}
    m^*(x) &= b + \mu(x) + \pi h(x),\\
    \widehat{e}_i &= Y_i - m^*(X_i) = (Z_i - \pi) h(X_i) + \eps_i.
\end{align*}
Then, it is easy to notice that the ideal RES is a special case of ML-FRT, by replacing $Y_i$ with $\widehat{e}_i$ and replacing $X_i$ with constant one. To apply Theorem~\ref{thm:power} to the ideal RES, it suffices to verify Assumptions~\ref{asmp:iid} and \ref{asmp:func_class} under the ideal RES setup. 

We first verify that the residuals $\widehat{e}_i$ are bounded. Under Assumption \ref{asmp:iid_res}, for any $i$, we have
$$
|\widehat{e}_i|\le |Z_i-\pi||h(X_i)| +  |\eps_i|\le |h(X_i)|+|\eps_i| \le M.
$$
Therefore, ``Assumption~\ref{asmp:iid}'' for ideal RES holds. 

Next, we verify Assumption~\ref{asmp:func_class} on function classes $\cF$ and $\cF_0$, that is,
\begin{enumerate}
    \item $\cF_0, \cF$ are uniformly bounded. 
    \item For a given best predictor $f_r^*\in \inf_{f\in\cF} \E(\widehat{e}_i - f(Z_i^{(r)}))^2$, we have $\E h(X_i) (f_r^*(1) - f_r^*(0)) \ge 0$.
    \item $\cR_n(\cF;\P), \cR_n(\cF_0;\P), \cR_n(\cF;\P^{(r)}) = o(1)$.  
\end{enumerate}
By construction, they are uniformly bounded by $M$ so that Condition 1 is satisfied. Moreover,
\begin{align*}
    \inf_{f\in\cF} \E(\widehat{e}_i - f(Z_i^{(r)}))^2 &= \inf_{\beta_1, \beta_2\in[-M/2, M/2]} \E ((Z_i-\pi)h(X_i) + \eps_i - \beta_1 - \beta_2 Z_i^{(r)})^2\\
    &= \inf_{\beta_1, \beta_2\in[-M/2, M/2]} \E ((Z_i-\pi)h(X_i) + \eps_i - \beta_1 - \beta_2 \pi - \beta_2 (Z_i^{(r)}-\pi))^2\\
    &= \inf_{\beta_1, \beta_2\in[-M/2, M/2]} \E ((Z_i-\pi)h(X_i) + \eps_i - \beta_1 - \beta_2 \pi)^2 + \E (\beta_2 (Z_i^{(r)}-\pi))^2\\
    &\ge \inf_{\gamma\in\R} \E ((Z_i-\pi)h(X_i) + \eps_i - \gamma)^2 = \var((Z_i - \pi)h(X_i) + \eps_i).
\end{align*}
Therefore, the optimal predictor $f_r^*$ on the randomized data satisfies $\beta_1 = \beta_2 = 0$. This implies
\begin{equation*}
    \E h(X) (f_r^*(1) - f_r^*(0)) = \beta_2 \E h(X)=0~.
\end{equation*}
Therefore, Condition 2 is satisfied. Lastly, given the linear function classes $\cF$ and $\cF_0$, it is easy to verify that their Rademacher complexity converges to zero \citep{wainwright2019high}. Thus, ``Assumption~\ref{asmp:func_class}'' for ideal RES holds. 

Since all assumptions required for Theorem~\ref{thm:power} are satisfied with outcomes $\widehat{e}_i$ and function classes $\cF$, $\cF_0$, Theorem~\ref{thm:power} holds for the ideal RES test with
\begin{equation*}
    \Delta = \inf_{f\in \cF} \E (\widehat{e} - f(Z^{(r)}))^2 - \inf_{f\in \cF} \E (\widehat{e} - f(Z))^2.
\end{equation*}

\noindent \underline{Step 2. Asymptotic equivalence between $m^*$ and $\widehat{m}$.}
Here, we show that the actual RES has the same type II error bound established above for the ideal RES. 
Based the proof of Theorem~\ref{thm:power} (Lemma~\ref{lem:power}), to analyze the actual RES approach, we need to bound
\begin{equation*}
    \P(\mathrm{pval}>\alpha) \le R\Bigl(\P\bigl(t_n^{\res}(\YY, \ZZ^{(r)}, \XX; \widehat{m})\ge c\bigr) + \P\bigl(t_n^{\res}(\YY, \ZZ, \XX; \widehat{m})\le c\bigr)\Bigr).
\end{equation*}
For brevity, let $\widehat{T}^{(r)}$ and $\widehat{T}$ denote the test statistic under fitted model $\widehat{m}$, and let $T^{(r)}$ and ${T}$ denote the test statistic under the best model ${m}^*$. We relate $\widehat{m}$ with $m^*$ to obtain 
\begin{align*}
    \P\bigl(\widehat{T}^{(r)}\ge c\bigr) + \P\bigl(\widehat{T}\le c\bigr) &= \P\bigl(\widehat{T}^{(r)} - T^{(r)} + T^{(r)}\ge c\bigr) + \P\bigl(\widehat{T} - T + T\le c\bigr) \\
    &\le \P\bigl(\widehat{T}^{(r)} - T^{(r)} \ge \delta_n\bigl) + \P\bigl(T^{(r)}\ge c-\delta_n\bigr)\\
    &+ \P\bigl(\widehat{T} - T \le -\delta_n\bigr) +  \P\bigl( T \le c+\delta_n\bigr).
\end{align*}
Then, to show the actual RES has the same power result, it suffices to prove 
\begin{align}
    \P(\widehat{T}^{(r)} - T^{(r)} \ge \delta_n) + \P(\widehat{T} - T \le -\delta_n) &= o(1), \label{eq:asymp_diff}\\
    \P(T^{(r)}\ge c-\delta_n) + \P( T \le c+\delta_n) &\asymp \P(T^{(r)}\ge c) + \P( T \le c)\label{eq:asymp_T}.
\end{align}
The equations above imply that the probability under $\widehat{m}$ is asymptotically equal to that under $m^*$. Together with Step 1, we know that the bound in Theorem~\ref{thm:power} holds for the RES with $\widehat{m}$.

To prove Equations~\eqref{eq:asymp_diff} and \eqref{eq:asymp_T}, we define
\begin{equation*}
    C_n = \sup_{x\in\cX} |\widehat{m}(x) - m^*(x)|,
\end{equation*}
which captures the difference between $\widehat{m}$ and $m^*$. 
Moreover, Assumption~\ref{asmp:consistent} indicates that $C_n = o(1)$. The following lemma provides a way to specify $\delta_n$ such that  Equations~\eqref{eq:asymp_diff} and \eqref{eq:asymp_T} hold.
\begin{lemma}\label{lem:consistent}
Under Assumption~\ref{asmp:consistent}, there exists a fixed constant $B$ such that with probability one, we have
\begin{equation*}
|\widehat{T}^{(r)} - T^{(r)}| \le B C_n, \quad |\widehat{T} - T| \le B C_n.
\end{equation*}
\end{lemma}
Based on Lemma~\ref{lem:consistent}, we may choose $\delta_n = (B + 1) C_n$, which guarantees that Equation~\eqref{eq:asymp_diff} holds. On the other hand, we have $\delta_n = o(1)$ since $C_n = o(1)$. From the proof of Lemma~\ref{lem:power}, when $\Delta >0$, the specified constant $c$ is of constant order. Therefore, by the right continuity of the distribution function, Equation~\eqref{eq:asymp_T} holds. When $\Delta = 0$, we get a trivial error bound, which clearly holds for the actual RES. To sum up, we show that the actual RES enjoys the same type II error bound in Step 1. 

\begin{proof}[Proof of Lemma~\ref{lem:consistent}]
By definition, we have
\begin{equation*}
    \CV(\widehat{\eps}_i \sim 1) = \frac{1}{k} \sum_{j=1}^k \frac{1}{|D_j|} \sum_{i\in D_j} \left(\widehat{\eps}_i - \widehat{\mu}_j\right)^2,
\end{equation*}
where $\widehat{\mu}_j$ solves the empirical risk minimization on $j$-th training set, i.e., 
\begin{equation*}
    \widehat{\mu}_j = \underset{\mu}{\arg\min} \sum_{i\in D_j^\complement} (\widehat{\eps}_i - \mu)^2 = \frac{1}{|D_j^\complement|} \sum_{i\in D_j^\complement} \widehat{\eps}_i.
\end{equation*}
Therefore, we obtain
\begin{equation*}
    \CV(\widehat{\eps}_i \sim 1) = \frac{1}{k} \sum_{j=1}^k \frac{1}{|D_j|} \sum_{i\in D_j} \left(\widehat{\eps}_i - \frac{1}{|D_j^\complement|} \sum_{l\in D_j^\complement} \widehat{\eps}_l\right)^2.
\end{equation*}
Therefore, we have
\begin{multline*}
    \CV(\widehat{\eps}_i \sim 1) - \CV(\widehat{e}_i \sim 1) = \frac{1}{k} \sum_{j=1}^k \frac{1}{|D_j|} \sum_{i\in D_j} \left[\Bigl(\widehat{\eps}_i - \frac{1}{|D_j^\complement|} \sum_{l\in D_j^\complement} \widehat{\eps}_l\Bigr)^2 - \Bigl(\widehat{e}_i - \frac{1}{|D_j^\complement|} \sum_{l\in D_j^\complement} \widehat{e}_l\Bigr)^2\right]\\
    = \frac{1}{k} \sum_{j=1}^k \frac{1}{|D_j|} \sum_{i\in D_j} \Bigl(\widehat{\eps}_i + \widehat{e}_i - \frac{1}{|D_j^\complement|} \sum_{l\in D_j^\complement} (\widehat{\eps}_l+ \widehat{e}_l) \Bigr) \Bigl(\widehat{\eps}_i - \widehat{e}_i - \frac{1}{|D_j^\complement|} \sum_{l\in D_j^\complement} (\widehat{\eps}_l - \widehat{e}_l) \Bigr).
\end{multline*}
Assumption~\ref{asmp:consistent} implies that with probability one, we have $|\widehat{\eps}_i - \widehat{e}_i| \le C_n$. In addition, Assumption~\ref{asmp:iid_res} suggests that $|\widehat{e}_i|\le M$. Therefore, one can find a positive constant $B$, such that $|\widehat{\eps}_i + \widehat{e}_i| \le B$ for all $i$. As a result, we have
\begin{align*}
    |\CV(\widehat{\eps}_i \sim 1) - \CV(\widehat{e}_i \sim 1)| &\le \frac{1}{k} \sum_{j=1}^k \frac{1}{|D_j|} \sum_{i\in D_j} (2B) (2 C_n) \\
    &\le 4 B C_n.
\end{align*}
We can apply a similar idea to show that 
\begin{align*}
    |\CV(\widehat{\eps}_i \sim 1 + Z_i) - \CV(\widehat{e}_i \sim 1 + Z_i)| &\le 4 B C_n, \\
    |\CV(\widehat{\eps}_i \sim 1 + Z_i^{(r)}) - \CV(\widehat{e}_i \sim 1 + Z_i^{(r)})| &\le  4 B C_n.
\end{align*}
The proof of the lemma is then complete.
\end{proof}

\noindent\underline{Step 3. Simplify $\Delta$.}
For $\Delta^{\res}$, as we show in Step 1, we have
\begin{equation*}
    \inf_{f\in\cF} \E (\widehat{e} - f(Z^{(r)}))^2 = \E ((Z-\pi)h(X) + \eps)^2 = \var(Z) \E h^2(X) + \var(\eps),
\end{equation*}
where the minimum is achieved with $\beta_1 = \beta_2 = 0$.
Similarly, we can show that
\begin{align*}
    \inf_{f\in\cF} \E (\widehat{e} - f(Z))^2 &= \inf_{\beta_1, \beta_2\in[-M/2, M/2]} \E (Z-\pi) h(X) - \beta_1 - \beta_2 Z)^2 + \var(\eps)\\
    &\ge \inf_{\beta_1\in\Real, \beta_2\in \Real} \E (Z-\pi) h(X) - \beta_1 - \beta_2 Z)^2 + \var(\eps)\\
    &\ge \inf_{\beta_2\in \Real} \E ((Z-\pi) (h(X) - \beta_2))^2 + \var(\eps)\\
    &= \inf_{\beta_2\in \Real} \var(Z) \E(h(X) - \beta_2))^2 + \var(\eps)\\
    &\ge \var(Z) \var(h(X)) + \var(\eps).
\end{align*}
Note that the lower bound is obtained with $\beta_1 = 0$ and $\beta_2 = \E h(X)\in[-M/2, M/2]$. Therefore, we have
\begin{equation*}
\Delta^{\res} = \var(Z) E h^2(X) - \var(Z) \var(h(X)) = \var(Z)[\E h(X)]^2 = \pi(1-\pi) [\E h(X)]^2.
\end{equation*}

Following the same line as in the analysis for $\Delta^{\res}$, we can show that under Assumption~\ref{asmp:func_class_inf}
\begin{equation*}
    \Delta^{\ml} = \pi(1-\pi) \E h^2(X),
\end{equation*}
which completes the proof.
\end{proof}

Next, we prove Proposition~\ref{prop:relative_efficiency}. 

\begin{proof}[Proof of Proposition~\ref{prop:relative_efficiency}]
First we analyze the deterministic test $\phi^{\ml}$ with respect to the ML-FRT. Recall that 
\begin{align*}
\phi^{\ml} = \mathbb{I}\{t_n(Y, Z, \bX) > q^{\ml}_{n,\alpha}\}.
\end{align*}
Here, $q_{n, \alpha}^{\ml}$ is the $1-\alpha$ quantile of $t_n(Y,Z,\bX)$ under $H_0^{\glob}: h = 0$. 
The proof consists of two steps. 
\begin{enumerate}
    \item Show that the critical value $q_{n,\alpha}^{\ml} = o(1)$. 
    \item Use the lower bound in Assumption~\ref{asmp:large_deviation} to control the type II error $\P(\phi^{ml} = 0)$. 
\end{enumerate}

\noindent\underline{Step 1. Analyze $q_{n,\alpha}^{\ml}$. }
Given $H_0^{\glob}:h = g = 0$, we have $Y_i = \mu + b(X_i)+\eps_i$. Then, under Assumptions~\ref{asmp:iid} and \ref{asmp:func_class}, we apply Step 1 in the proof of Lemma~\ref{lem:power} to obtain
\begin{align*}
    \left|t_n(\YY, \ZZ, \XX) - \Delta_0\right| &\le 4\max\{\cR_{n-n/k}(\cF_{0};\P), \cR_{n-n/k}(\cF;\P)\} + 4t, \\
    \Delta_0 &\coloneqq \E l(f_{0}^*) - \E l(f^*),
\end{align*}
with probability at least $1 - 4k \exp\left(-\frac{2 n t^2}{k M_0^2}\right) - 2\exp\left(-\frac{(k-1)nt^2}{2kM_0^2}\right)$. Since $\cR_{n-n/k}(\cF_{0};\P), \cR_{n-n/k}(\cF;\P)$ are $o(1)$, the above result implies that $t_n(Y,Z,\bX)$ converges in probability to $\Delta_0$. Given $Y_i = \mu + b(X_i)+\eps_i$, we write
\begin{align*}
    \E l(f_{0}^*) &= \inf_{f_0\in\cF_0} \E(\mu + b(X)+\eps - f_0(X))^2 = \inf_{f_0\in\cF_0} \E(\mu + b(X) - f_0(X))^2 + \E \eps^2,\\
    \E l(f^*) &= \inf_{f\in\cF} \E (\mu + b(X)+\eps - f(X,Z))^2 = \inf_{f\in\cF} \E (\mu + b(X) - f(X,Z))^2 + \E \eps^2.
\end{align*}
Then, one can verify that the infimum is obtained at $\mu+b(x)$ for $\E l(f_{0}^*)$ and $\E l(f^*)$, as defined in Proposition~\ref{prop:relative_efficiency}. This implies $\Delta^0 = 0$, and hence $t_n(Y,Z,\bX)$ converges in probability to zero. Therefore, the $1-\alpha$ quantile of $t_n(Y,Z,\bX)$, i.e., $q_{n,\alpha}^{\ml}$, also converges to zero. 

\noindent\underline{Step 2. Analyze type II error.}
Under the alternative, we analyze the limiting behavior of 
\begin{equation*}
\P(\phi^{\ml} = 0) = \P(t_n(Y, Z, \bX) < q^{\ml}_{n,\alpha})= \P(t_n(Y, Z, \bX) - \Delta^{\ml}< q^{\ml}_{n,\alpha} - \Delta^{\ml}).
\end{equation*}
By the assumption that $|\E h(X)| > 0$, we have $\E h^2(X) \ge |\E h(X)|^2 > 0$, and hence $\Delta^{\ml} > 0$.
By Step 1, for any $\eps>0$, we have $|q_{n, \alpha}^{\ml}|<\eps$ for $n$ large enough. Since $-\eps < q_{n, \alpha}^{\ml}$, we have
\begin{align*}
    \underbrace{\P(t_n(Y, Z, \bX) - \Delta^{\ml}< - \eps - \Delta^{\ml})}_\text{(I)} 
    &\le \P(t_n(Y, Z, \bX) - \Delta^{\ml}< q^{\ml}_{n,\alpha} - \Delta^{\ml}).
\end{align*}
By Assumption~\ref{asmp:large_deviation}, we have
\begin{align*}
    &\lim \frac{1}{n}\log \text{(I)} \ge - I(\eps + \Delta^{\ml}).
\end{align*}
Therefore, we have
\begin{align*}
    \lim\inf \frac{1}{n}\log \P(t_n(Y, Z, \bX) - \Delta^{\ml}< q^{\ml}_{n,\alpha} - \Delta^{\ml}) \ge - I(\eps + \Delta^{\ml}).
\end{align*}
Since the inequality above holds for any $\eps > 0$ and $I(x)$ is continuous at point $\Delta^{\ml}$, we have
\begin{equation}\label{eq:large_dev1}
    \lim\inf \frac{1}{n}\log\P(t_n(Y, Z, \bX) - \Delta^{\ml}< q^{\ml}_{n,\alpha} - \Delta^{\ml}) \ge - I(\Delta^{\ml}).
\end{equation}

For the $\phi^{\res}$, we can apply similar arguments as Steps 1 and 2 above. First, under the null hypothesis and Assumption~\ref{asmp:consistent}, $\widehat{m}$ is asymptotically equivalent to $m^*(x) = \mu + b(x)$. Additionally, for $t^{\res}(Y,Z,\bX;m^*)$, we can apply the same analysis as in the proof of Lemma~\ref{lem:power} to show that $t^{\res}(Y,Z,\bX;m^*)$ concentrates around zero. Therefore, $q_{n,\alpha}^{\res}$ converges to zero. 

Second, we apply Assumption~\ref{asmp:large_deviation} as in Step 2 above to obtain 
\begin{equation}\label{eq:large_dev2}
    \underset{n\to\infty}{\lim\sup} \frac{1}{n} \log \P(\phi^{\res} = 0) \le - I(\Delta^{\res})~.
\end{equation}
Lastly, taking difference of Equations~\eqref{eq:large_dev1} and \eqref{eq:large_dev2} gives the final result.
\end{proof}

\subsection{Upper Bounds on Rademacher Complexity}\label{sec:rademacher}
In Theorem~\ref{thm:power}, we assume that $\cR_{n-n/k}(\cF_0;\P)$, $\cR_{n - n/k}(\cF;\P)$, and $\cR_{n - n/k}(\cF;\P^{(r)})$ are $o(1)$. Here we justify this assumption by deriving concrete upper bounds for Rademacher complexities.

Our upper bound relies on a notion of pseudo-dimension defined below. 
\begin{definition}
Let $\mathcal{G}$ be a collection of real-valued functions defined on a set $\mathcal{Z}$. Given a subset $S:=\left\{z_1, \ldots, z_m\right\} \subset \mathcal{Z}$, we say $S$ is pseudo-shattered by $\mathcal{G}$ if there are $r_1, \ldots, r_m \in \mathbb{R}$ such that for each $b \in\{0,1\}^m$ we can find $g_b \in \mathcal{G}$ satisfying $\operatorname{sign}\left(g_b(z_i)-r_i\right)=b_i$ for all $i \in[m]$. We define the pseudo-dimension of $\mathcal{G}$, denoted as $\mathrm{Pdim}(\mathcal{G})$, as the maximum cardinality of a subset $S \subset \mathcal{Z}$ that is pseudo-shattered by $\mathcal{G}$.
\end{definition}
Similar to the Rademacher complexity, pseudo-dimension serves as a complexity measure for general function classes. However, we highlight that the pseudo-dimension does not depend on sample size $n$, and therefore can be viewed as a ``constant'' in our analysis. 

\begin{proposition}\label{prop:rademacher}
Under Assumptions~\ref{asmp:iid} and \ref{asmp:func_class}, we have
\begin{align*}
    \cR_{n}(\cF;\P), \cR_{n}(\cF;\P^{(r)}) = O\left(\sqrt{\frac{\log n \times \mathrm{Pdim}(\cF)}{n}}\right),\quad \cR_{n}(\cF_0;\P) = O\left(\sqrt{\frac{\log n \times \mathrm{Pdim}(\cF_0)}{n}}\right)~.
\end{align*}
\end{proposition}
This result suggests that the decay of Rademacher complexities depend on the pseudo-dimension of the function classes. In particular, when $\cF_1$ is parametrized by neural network classes, tight pseudo-dimension bounds in \cite{anthony1999neural,bartlett2019} can be plugged in Proposition~\ref{prop:rademacher} for concrete bounds that diminishes to zero. 

\begin{proof}
Here we give the proof for $\cR_{n}(\cF;\P)$, as the proof for other complexities is almost identical. 

Recall that
\begin{equation*}
    \cR_n(\cF;\P) \coloneqq \underset{(Y, \mathbf{X}, Z)}{\E} \underset{\sigma}{\E} \Biggl(\sup_{f\in\cF} \underbrace{\left| \frac{1}{n} \sum_{i=1}^n \sigma_i (Y_i - f(X_i, Z_i))^2\right|}_\text{$\eqqcolon X_f$}\Biggr).
\end{equation*}
First, observe that Assumptions~\ref{asmp:iid} and \ref{asmp:func_class} implies $l(f) \le 4M^2 \eqqcolon M_0$, and hence
\begin{align*}
    |l(Y, X, Z; f) - l(Y, X, Z;f')| &\le |(Y - f(X, Z))^2 - (Y - f'(X,Z))^2 | \\
    &\le (|Y - f(X, Z)| + |Y - f'(X,Z)| )|f(X, Z) - f'(X,Z)|\\
    &\le 2\sqrt{M_0} |f(X, Z) - f'(X,Z)|.
\end{align*}
Therefore, for $X_f$ defined above, we have
\begin{align*}
    |X_f - X_{f'}| &\le \frac{1}{n} \sum_{i=1}^n |l(Y_i, X_i, Z_i; f) - l(Y_i, X_i, Z_i;f')| \\
    &\le \frac{1}{n} \sum_{i=1}^n 2\sqrt{M_0} |f(X_i, Z_i) - f'(X_i,Z_i)|\\
    &\le  2\sqrt{M_0} \max_{i} |f(X_i, Z_i) - f'(X_i,Z_i)| \eqqcolon \rho(f, f').\\
\end{align*}

For $\epsilon>0$, let $\calN_{\infty}\left(\epsilon, \cF,\{X_i, Z_i\}_{i=1}^n\right)$ be the minimal $\epsilon$-covering net of $\cF_1$ under the pseudometric $d$ induced by $\{X_i, Z_i\}_{i=1}^n$:
$$
d\left(f, f^{\prime}\right):=\max _{i\in[n]}\left|f\left(X_i, Z_i\right)-f'\left(X_i, Z_i\right)\right|.
$$
In other words, for any $f \in \cF$, we can find $f^{\prime} \in \calN_{\infty}\left(\epsilon, \cF,\{X_i, Z_i\}_{i=1}^n\right)$ such that $d\left(f, f^{\prime}\right) \leq \epsilon$. To simplify the notation, we use $\calN_{\infty} = \calN_{\infty}\left(\epsilon, \cF,\{X_i, Z_i\}_{i=1}^n\right)$. Additionally, we define the uniform covering number as 
\begin{equation*}
    N_{\infty, n}:= \sup \left\{|N_{\infty}\left(\epsilon, \mathcal{F},\left\{X_i, Z_i\right\}_{i=1}^n\right)|: (X_1,Z_1), \ldots, (X_n, Z_n) \in \cX\times \cZ \right\}.
\end{equation*}

Given $\epsilon>0$, for any $f \in$ $\cF_1$, we can find $f'$ such that
$$
\rho(f, f') \le \eta \epsilon,
$$
where $\eta=2 \sqrt{M_0}$. As a result, one can easily check
\begin{align*}
    \sup _{f\in\cF} X_{f} & \le \sup _{\rho\left(f, f'\right) \le \eta \epsilon}\left|X_{f}-X_{f'}\right|+\sup _{f \in \calN_{\infty}} X_{f} \\
    & \le \eta \epsilon+\sup _{f \in \calN_{\infty}} X_{f}.
\end{align*}
Note that $X_f$ is the absolute value of a sub-Gaussian random variable (with respect to $\sigma$) with a sub-Gaussian parameter $v = M_0^2/4n$ \citep{boucheron2013concentration}. Hence, the maximal inequality \citep[Section 2.5]{boucheron2013concentration} yields
\begin{equation*}
    \E_{\sigma}\sup _{f\in\cF} X_{f} \le \eta \epsilon + \sqrt{\frac{M_0^2\log(2|\calN_{\infty}|)}{2n}} \le \eta \epsilon + C\sqrt{\frac{\log|\calN_{\infty}|}{n}}
\end{equation*}
for some large enough constant $C>0$. Then we have
\begin{equation*}
    \underset{(Y, \mathbf{X}, Z)}{\E} \underset{\sigma}{\E} \sup _{f\in\cF} X_{f} \le \eta \epsilon + C\sqrt{\frac{\log\calN_{\infty, n}}{n}},
\end{equation*}
by the definition of the uniform covering number.

Lastly, we reduce the uniform covering number to the pseudo-dimension based on Theorem 12.2 of \cite{anthony1999neural}
\begin{equation*}
    \calN_{\infty, n}\leq\left(\frac{2 e n M}{\epsilon \cdot \operatorname{Pdim}\left(\cF\right)}\right)^{\mathrm{Pdim}\left(\cF\right)}.
\end{equation*}
Hence, 
\begin{align*}
\underset{(Y, \mathbf{X}, Z)}{\E} \underset{\sigma}{\E} \sup _{f\in\cF} X_{f} &= O\left(\eta \epsilon + \sqrt{\frac{\mathrm{Pdim}(\cF)\log(\frac{n}{\epsilon})}{n}}\right)\\
&= O\left(\sqrt{\frac{\mathrm{Pdim}(\cF)\log(n)}{n}}\right),
\end{align*}
where the last line comes from choosing $\epsilon = 1/\sqrt{n}$.
\end{proof}

\section{Validity of $\mathrm{pval}^{\mathrm{het},\gamma}_n$}\label{het_validity}
\begin{theorem}\label{thm:het}
Suppose that $H_0^\het$ holds true. Then,
\begin{equation*}
    \P\big(\mathrm{pval}^{\mathrm{het},\gamma}_n \leq \alpha) \le \alpha,~\text{for any $\alpha \in [0,1]$ and any $n>0$},
\end{equation*}
where the randomness in $\P$ is with respect to the 
experimental design $\Pn$.
\end{theorem}
\begin{proof}
The proof closely follows the proof of Lemma in \citet[Section 2]{berger1994pvalues}. Let $\mathrm{CI}_{\gamma}$ satisfy $\P(\tau_0 \in \mathrm{CI}_\gamma) \geq 1-\gamma$ and $\tau^*$ denote the true constant treatment effect. If $\gamma > \alpha$, the result holds trivially. If $\gamma \leq \alpha$, we have
\begin{align*}
    \P\big(\mathrm{pval}^{\mathrm{het},\gamma}_n \leq \alpha\big) &= \P\big(\mathrm{pval}^{\mathrm{het},\gamma}_n \leq \alpha, \tau^* \in \mathrm{CI}_\gamma\big) + \P\big(\mathrm{pval}^{\mathrm{het},\gamma}_n \leq \alpha, \tau^* \notin \mathrm{CI}_\gamma\big) \\
    &\leq \P\big(\mathrm{pval}^{\mathrm{het},\gamma}_n \leq \alpha, \tau^* \in \mathrm{CI}_\gamma\big) + \P\big(\tau^* \notin \mathrm{CI}_\gamma\big) \\
    &\leq \P\big(\mathrm{pval}(\tau^*) + \gamma \leq \alpha\big) + \gamma \\
    &\leq \alpha - \gamma + \gamma = \alpha,
\end{align*}
where the second inequality comes from $\sup_{\tau_0 \in \mathrm{CI}_\gamma}\mathrm{pval}(\tau_0) \geq \mathrm{pval}(\tau^*)$ when $\tau^* \in \mathrm{CI}_\gamma$.
\end{proof}

\section{Sample Size Determination Example}\label{sec:sample_size_example}
We illustrate the sample size calculation in Section~\ref{sec:sample_size} using a simulation study.
Consider the data generating process in Section \ref{sec:num_validation}, which is stated again below for convenience. We set $n = 100$, $p = 2$, and $X_i \stackrel{iid}{\sim} \calN(0, \Sigma)$, where $\Sigma$ is a randomly generated correlation matrix based on the R package \texttt{randcorr}. We specify model \eqref{eq:our_model} by setting $b(x) = 0.1 x^\top\beta$, random coefficients $\beta\sim\mathrm{U}([1, 5]^p)$, $g = 0$, 
and
\begin{equation*}
h(x) = \begin{cases}
    \tau + \tau \min \{\frac{2}{x_1}, 10\} &~\text{if}~x_1>0 \\
    \tau + \tau \max \{\frac{2}{x_1}, -10\} &~\text{if}~x_1< 0 \\
    \tau  &~\text{otherwise},
\end{cases},
\end{equation*}
where $\epsilon_i \sim\calN(0, 0.1^2)$. The experimental design is an i.i.d. Bernoulli design with treatment probability 0.5.

Suppose we want to determine the necessary sample size under the alternative $\tau = 1$ to achieve 80\% power. Using the simulated data and the procedure in Section~\ref{sec:sample_size} based on random forests, we obtain that $\widehat{L}=4.98$, $k = 10$, and $\widehat{M}_0 = 9.98$. Then, omitting the constant $R$ as suggested, we may solve for $n$ in the Equation~\eqref{eq:sample_size} to obtain $n \approx 7,770$. 

\end{document}